\documentclass[pra,aps,amsmath,amssymb,superscriptaddress,twocolumn,longbibliography]{revtex4-1}
\usepackage{graphicx,multirow}
\usepackage{color}

\newcommand{\half}{\frac{1}{2}}

\newcommand{\uotimes}{~\underline{\otimes}~}
\newcommand{\votimes}{~\overline{\otimes}~}
\newcommand{\votime}{~\overline{\otimes}_1~}
\newtheorem{thm}{Theorem}

\newtheorem{definition}{Definition}

\newenvironment{proof}[1][Proof.]{\begin{trivlist}
\item[\hskip \labelsep {\bfseries #1}]}{\end{trivlist}}

\newenvironment{example}[1][Example.]{\begin{trivlist}
\item[\hskip \labelsep {\bfseries #1}]}{\end{trivlist}}

\newcommand{\qed}{\nobreak \ifvmode \relax \else
      \ifdim\lastskip<1.5em \hskip-\lastskip
      \hskip1.5em plus0em minus0.5em \fi \nobreak
      \vrule height0.75em width0.5em depth0.25em\fi}
\begin{document}

\title{On the origin of nonclassicality in single systems}
\author{S.            Aravinda}           \email{aravinda@imsc.res.in}
\affiliation{Institute of Mathematical Sciences, Chennai, India}
%Poornaprajna   Sadashivnagar,  Bengaluru}  
\author{R.   Srikanth}
\email{srik@poornaprajna.org}  \affiliation{Poornaprajna Institute  of
  Scientific  Research,   Sadashivnagar,  Bengaluru-   560080,  India}
\author{Anirban        Pathak}       \email{anirban.pathak@jiit.ac.in}
\affiliation{Jaypee   Institute  of   Information  Technology,   A-10,
  Sector-62, Noida, UP-201307, India}

\begin{abstract}
In the  framework of  certain general  probability theories  of single
systems,   we  identify   various   nonclassical   features  such   as
incompatibility,  multiple   pure-state  decomposability,  measurement
disturbance,  no-cloning and  the impossibility  of certain  universal
operations, with  the non-simpliciality of  the state space.   This is
shown to  naturally suggest  an underlying  simplex as  an ontological
model.   Contextuality turns  out  to be  an independent  nonclassical
feature, arising from the intransitivity of compatibility.
\end{abstract}

\pacs{03.65.Ta,03.65.Ud,  03.67.Dd}  

\keywords{Foundations   of  quantum   mechanics,  convex   operational
  theories, single system, nonclassicality}

% 03.65.Ta	Foundations of quantum mechanics;
% 03.67.Dd	Quantum cryptography and communication security
% 03.67.Hk	Quantum communication
% 03.65.Ud	Entanglement and quantum nonlocality

\maketitle

\section{Introduction}

 It is  an interesting  question in  quantum foundations  whether
quantum theory can be derived from information theoretic principles, a
program    that    may    be   traced    to    \cite{fuchs2001quantum,
  fuchs2003quantum,  brassard2005information} and  which has  received
considerable        attention       recently        \cite{Har01,CDP11,
  hardy2011reformulating, MM11, masanes2013existence}.

It is  well-known that nonclassicality in  multipartite systems arises
from nonlocal  non-signaling correlations  \cite{Bel64}, an  area that
has  been  extensively studied  from  a  foundational and  information
theoreitc  perspective  \cite{Bar07,  BLM+05,  BBL+06,  BM06,  BBL+07,
  Bar03, JH14}.

In particular,  various earlier works  have linked nonlocality  in the
context  of \textit{multipartite}  systems to  nonclassical properties
such   as  intrinsic   randomness  \cite{MAG06},   non-comeasurability
\cite{BGG+0,*Ban15,   SAB+05,   WGC09,   QVB14,   SB14},   uncertainty
\cite{OW10},   local   quantum  mechanical   features   \cite{BBB+10},
non-simpliciality of the state space \cite{BBL+07}. Yet, there appears
to be  no systematic study of  all these properties solely  within the
context of single  systems. It is this issue that  is addressed in the
present work.

We    call   the    above-mentioned   single-system    properties   as
\textit{base-level nonclassical}.  These  nonclassical features can be
reproduced by an epistemic model over  ontic states, as first shown by
Spekkens \cite{Spe07}.  By  contrast, the lack of  a joint probability
distribution   (JD)   over   measurement   outcomes   represents   the
complementary, ``top-level'' nonclassicality.

% As such, our  approach is an attempt to isolate  the most elementary
% surprise about a nonclassical theory.
Now, contextuality and nonlocality are  closely connected. Yet, to our
knowledge there are no works that relate contextuality to nonclassical
properties like incompatibility or no-cloning.   It is this issue that
is  addressed  here.   A  related   but  distinct  issue  is  that  of
identifying   a  simple   principle  to   derive  the   quantum  bound
\cite{Cab13,AFL+15,Hen12,Hen15,CSW14} on contextuality \cite{KS67}.

% Interestingly, our characterization of nonclassicality as associated
% with  the non-simpliciality  of the  state space  suggests a  novel,
% geometric approach to the ontology of nonclassical theories, wherein
% the ontological model is characterized  as an underlying simplex, of
% exponentially larger dimensionality for contextual theories.

In our  approach, we regard  a single system  as a composite  built by
aggregating    individual    \textit{properties},    represented    as
observables. In  practice, we shall  regard any given observable  as a
specific measurement  strategy.  The  basic idea  is that  a classical
system  is  one where  properties  coexist  peacefully, whereas  in  a
nonclassical  system, there  is some  sort of  ``tension'' prohibiting
their  peaceful  coexistence,  or \textit{congruence}.   Tension  here
corresponds   to   the  statement   that   the   state  space   of   a
finite-dimensional  single-system  isn't  a  simplex,  manifesting  as
incompatibility, no-cloning, etc.  Our  approach employs the framework
of  generalized   probability  theories  \cite{mackey2004mathematical,
  ludwig2012foundations,  BB93, Har01,  Har03,  Man03, BBL+07,  Bar07,
  Spe05, BW11, hardy2011foliable, CDP10}, with states and measurements
as  the basic  elements  of the  theory.  We  identify  two levels  of
nonclassicality:   the  above   mentioned  pairwise   tension  between
properties,   and  a   higher  level   nonclassicality  arising   from
``frustration'' in the graph of  congruence relations, leading to lack
of a joint probability distribution (JD) for the measurements.

In  Section  \ref{sec:frame}  we  formulate  our  framework  to  study
monopartite  systems,  identifying  the axiomatic  requirement  for  a
theory   to   be   classical   (Section   \ref{sec:classicality})   or
non-classical   (Section   \ref{sec:nonclassicality}).    Nonclassical
features that are consequences  of the nonsimpliciality assumption are
studied  in Section  \ref{sec:conseq}, among  them non-comeasurability
(Section    \ref{sec:nonco}),    measurement   disturbance    (Section
\ref{sec:disturb}), no-cloning  (Section \ref{sec:nocloning})  and the
impossibility  of  jointly  distinguishing all  pure  states  (Section
\ref{sec:jodi}).   The sense  in which  uncertainty also  follows from
nonsimpliciality is explained in Section \ref{sec:uncertain}.

Section   \ref{sec:simplionto}   explores  certain   consequences   of
nonsimpliciality in  conjunction with ontological  considerations. The
concept of an ontological model as a simplex underlying an operational
nonclassical theory  is developed in Section  \ref{sec:sqcup}. This is
then  applied  to  derive  a result  concerning  impossible  universal
operations   (\ref{sec:impossible})   and  preparation   contextuality
(\ref{sec:prep}).  This  Section thus  sets forth  what we  consider a
geometric approach to  the ontology of a simple  theory of monopartite
systems.

The above  enumerated sections  deal with  lower-level nonclassicality
arising from pairwise incongruence.  This type of nonclassicality will
be present even if the  pairwise incongruence is transitive.  However,
when pairwise  incongruence is intransitive, the  congruence structure
associated with  the properties  of the  single system  is nontrivial,
leading to states that lack a JD and thus correspond to contextuality,
as explored in Section  \ref{sec:JD}.  Section \ref{sec:cbox} develops
the  ontological  concept  of  an underlying  simplex  for  contextual
correlations.  Finally, we conclude in Section \ref{sec:conclu}.

\section{Regular theories for monopartite systems\label{sec:frame}}

 In   the   framework  of   generalized   probability   theories
\cite{mackey2004mathematical,   ludwig2012foundations,  BB93,   Har01,
  Har03, Man03, BBL+07, Bar07,  Spe05, BW11, hardy2011foliable, CDP10,
  chiribella2017quantum},  an  operational  theory  $\mathcal{T}$  for
monopartite systems is characterized by  the state space $\Sigma$, the
convex set of  all possible states associated to a  system governed by
the theory.  The set of extreme (or, pure)  points of $\Sigma$ is
denoted $\partial\Sigma$, taken to be  a finite number for our purpose
(so that $\mathcal{T}$ is discrete).

We shall define a  \textit{sharp measurement} as one that doesn't
admit  further   refinement,  i.e.,   the  measurement   shouldn't  be
expressible as a  convex combination of any other  measurements in the
theory  \cite{chiribella2016bridging}.   In quantum  mechanics,  sharp
measurments  correspond to  (non-degenerate) projective  measurements.
Suppose  that  measuring  sharp  measurement  $A$  on  state  $s$  and
post-selecting on outcome $a$ prepares  a pure state $\psi^A_a$, i.e.,
identifies some  state without any  error.  This can  be operationally
verified by repeated measurement of  $A$ on $\psi^A_a$, which does not
alter this state.  We call  such $\psi^A_a$ the \textit{eigenstate} of
the   sharp    measurement   $A$   with   ``eigenvalue''    $a$   (cf.
\cite{chiribella2015operational}).  In general,  not every measurement
has an  eigenstate.  We call a  theory $\mathcal{T}$ \textit{regular},
if every sharp measurement allowed  in $\mathcal{T}$ has an eigenstate
$\psi \in \partial\Sigma$.   Thus, a theory is irregular  if there are
sharp measurements that don't have eigenstates.  A canonical irregular
theory is the gdit theory, discussed in Section \ref{sec:gdit}.

Given  a set  of ``fiducial  measurements'' $A_j$,  \textit{tomographic
  separability}, a state $\psi \in  \Sigma$ is completely specified by
the  numbers  $P(a_j|A_j)$,  i.e.  by  the  conditional  probabilities
generated by the fiducial  measurements individually.  Else, the state
is  tomographically  nonseparable for  this  set.   This  may  be
considered as  the single-system  equivalent of  the axiom  of ``local
tomography''   \cite{barnum2014local}   or  ``tomographic   locality''
\cite{hardy2011reformulating},   encountered  in   reconstructions  of
quantum  theory.  The  minimal  set of  such  numbers $P(a|A)$  to
characterize   an    arbitrary   mixed    state   is    the   theory's
\textit{tomographic  dimension} $D$,  or  simply ``dimension''.   
Note that  this is  the same as  the number of  degrees of  freedom in
Hardy's axiomatic formulation of quantum theory \cite{Har01}.

\section{Classical theory\label{sec:classicality}}

A classical  theory $\mathcal{T}$  is a regular  theory where  all $m$
fiducial  measurements  $X_1,  X_2, \cdots,  X_m$  coexist  peacefully.
Mathematically,  we  express  peaceful  coexistence  as  follows.   We
associate   individual   state  spaces   $\Sigma_1,   \Sigma_2,\cdots,
\Sigma_m$ with the measurements, where $\Sigma_j$ is the convex hull of
all possible eigenstates of $X_j$.   Clearly, $\Sigma_j$ is a simplex,
meaning  that the  $n$ pure  points corresponding  to definite  values
(i.e.   measurement outcomes)  of measurement  $\Sigma_j$ are  linearly
independent.

Peaceful  coexistence  happens  precisely  if  the  full  state  space
$\Sigma$ is the \textit{convex direct product}
\begin{equation}
\Sigma \equiv \Sigma_1  \uotimes  \Sigma_2  \uotimes \cdots  \uotimes
\Sigma_m
\label{eq:tensorprod}
\end{equation}
of the individual state spaces $\Sigma_j$,  which is defined to be the
convex hull of the set of $n^m$ points:
\begin{equation}
\partial\Sigma  =  \partial\Sigma_1 \otimes  \partial\Sigma_2  \otimes
\cdots \otimes \partial\Sigma_m
\end{equation}
Since  these  points are  linearly  independent,  $\Sigma$ is  also  a
simplex and has dimension:
\begin{equation}
D = |\partial\Sigma|-1.
\label{eq:partial}
\end{equation}
Quite generally, we shall identify classicality with the simpliciality
of the state space $\Sigma$.

\begin{example}
Consider   the  classical   theory  $\mathcal{T}$   of  two   fiducial
measurements $X$  and $Z$  which take  values $x,  z \in  \{0,1\}$.  In
\textit{barycentric  coordinates}  \cite{GBC},   the  basis  for  both
$\Sigma_X$  and  $\Sigma_Z$  is  $\{(1,0),  (0,1)\}$.  The full  state space
$\Sigma_{XZ} \equiv  \Sigma_X \uotimes \Sigma_Z$  and is given  as the
convex hull of the points $\{(1,0), (0,1)\} \otimes \{(1,0), (0,1)\}$:
\begin{eqnarray}
|xz{=}00) &\equiv  (1,0,0,0),\nonumber\\
|xz{=}01) &\equiv (0,1,0,0), \nonumber\\
|xz{=}10) &\equiv (0,0,1,0),\nonumber\\
|xz{=}11) &\equiv (0,0,0,1).
\label{eq:ontixtate}
\end{eqnarray}
The general  mixed state is  an arbitrary point in  $\Sigma_{XZ}$, and
has  the   unique  barycentric  coordinate  representation   given  by
$(a,b,c,d)$   with    $a+b+c+d=1$   and    $a,b,c,d\ge0$.    Dimension
dim($\Sigma_{XZ}) =  3$. Thus, $\Sigma_{XZ}$  can be represented  as a
3-simplex (tetrahedron)  in $\mathbb{R}^3$.  Since the  four points in
Eq. (\ref{eq:ontixtate}) exist in  a three-dimensional space, they can
be  given  a  (not  unique)  Cartesian  representation,  e.g.  as  the
\textit{convex   independent}  points   $(0,0,0),  (0,1,1),   (1,0,1),
(1,1,0)$.  \hfill \qed
\end{example}

We  identify nonclassicality  with the  departure of  $\Sigma$ from  a
simplicial structure, as discussed below.

\section{Nonclassicality \label{sec:nonclassicality}}

Suppose  $\mathcal{T}$  is  a  regular  theory  characterized  by  two
fiducial  $n$-output  measurements  $X$  and $Z$,  whose  outcomes  are
denoted $x$ and $z$,  respectively. Here, peaceful coexistence between
$X$   and   $Z$  means   that   they   can  assume   definite   values
simultaneously. Thus conjunctive propositions (``$x$ and $z$'') can be
assigned a definite truth value (``true'' or ``false'').

If  $\mathcal{T}$  is  nonclassical,  $X$   and  $Z$  do  not  coexist
peacefully.  Loss of  peaceful coexistence means that $X$  and $Z$ are
linked disjunctively,  whereby propositions  (``$x$ and  $z$'') cannot
always be assigned  a definite truth value. In the  extreme case, only
propositions ``$x$ xor $z$'' have a definite truth value, even in pure
states.  This has the following geometric consequence.

Denote  the  $n$  eigenstates  of  $X$ by  $|\psi_X^j)$  and  the  $n$
eigenstates  of  $Z$  by  $|\psi_Z^k)$.  Therefore,  the  state  space
$\Sigma$  is  the  convex  hull of  these  $|\partial\Sigma|=2n$  pure
points.   Any state  in  $\Sigma$  is a  list  of $(n-1)$  conditional
probabilities   $P(x|X)$   and   $(n-1)$   conditional   probabilities
$P(z|Z)$. Therefore, the  dimension of $\Sigma$ is  $D \equiv 2(n-1)$.
The fact that $|\partial\Sigma| > D+1$  means that not all pure states
are linearly independent.  Therefore, $\Sigma$ is non-simplicial.

We   shall  refer   to   the  relations   between   the  elements   of
$\partial\Sigma$, that leads  to their linear dependence  and hence to
the   non-simpliciality  of   $\Sigma$,   as  the   ``nonsimpliciality
conditions''.  Typically, such a condition takes the form:
\begin{equation}
\sum_j q_j|\psi_X^j) = \sum_k r_k|\psi_Z^k),
\label{eq:constraint}
\end{equation}
with  $\sum_j q_j  =  \sum_k r_k=1$,  i.e., in  terms  of equality  of
certain   $X$-mixtures    and   $Z$-mixtures,   and   so    on.    Eq.
(\ref{eq:constraint}) has the interpretation  that two mixtures on the
LHS   and  RHS   are  operationally   indistinguishable.   Thus,   the
nonsimplicality conditions  have the  operational meaning  of multiple
pure-state   decompositions. 

Going beyond two  measurements, suppose a regular theory  has more than
two fiducial measurements,  namely $X, Y, Z, \cdots$.  It  is useful to
define the following
\begin{definition}
\label{defn:congruence}
The \textit{associated state space} of  a pair of measurements (say $X$
and $Z$), denoted $\Sigma_{XZ}$, is the convex hull of the eigenstates
of  $X$  and  $Z$.   We  describe  $X$  and  $Z$  as  \textit{pairwise
  congruent} precisely if $\Sigma_{XZ}$ is a simplex.  \hfill \qed
\end{definition}

A \textit{simple} nonclassical  theory of $m$-input-$n$-output systems
is a  noncontextual theory in  which all the $m$  fiducial measurements
are pairwise incongruent.  To each  fiducial measurement we associate a
normalized probability  distribution of  $(n-1)$ free  parameters, and
the state  space has dimension  of $D=m(n-1)$.  Because  
the property associated with each fiducial
measurement fails  to peacefully coexist  with any other,  we associate
precisely $n$  pure states  per fiducial measurement,  corresponding to
states that assigned a definite value under the measurement.  Therefore, there
will be $|\partial\Sigma|=mn$ pure states, and:
\begin{equation}
|\partial\Sigma|-D =m,
\label{eq:ah}
\end{equation}
from which the nonsimplicality of $\Sigma$ follows, for $m>1$.

The idea  that incongruence  causes properties that  are conjunctively
connected to become disjunctively  linked, is reflected mathematically
in the fact that the tensor product structure (\ref{eq:tensorprod}) is
replaced by one with tensor \textit{sum}:
\begin{equation}
\Sigma = \Sigma_1 \oplus \Sigma_2 \oplus \cdots \oplus \Sigma_m,
\label{eq:tensorsum}
\end{equation}
which captures  the idea  that the  dimensionalities are  added rather
than multiplies when composing systems.

\begin{example}
In the two-input-two-output operational theory $\mathcal{T}$, with two
dichotomic   measurements   $X$  and   $Z$,   suppose   the  $X$-   and
$Z$-eigenstates, which are the pure states of space $\Sigma$, are:
\begin{eqnarray}
\psi_X^+ &\equiv (1,0~|~\frac{1}{2},\frac{1}{2}) \nonumber \\ 
\psi_X^- &\equiv (0,1~|~\frac{1}{2},\frac{1}{2}) \nonumber \\
\psi_Z^+ &\equiv (\frac{1}{2},\frac{1}{2}~|~1,0) \nonumber \\
\psi_Z^- &\equiv (\frac{1}{2},\frac{1}{2}~|~0,1),
\label{eq:uneq}
\end{eqnarray}
where the vertical bar separates  the measurement probabilities of the
fiducial measurements $X$ and $Z$. The nonsimpliciality condition is:
\begin{equation}
\half(\psi_X^++\psi_X^-) = \half(\psi_Z^+ + \psi_Z^-)  = 
(\half,\half \mid \half, \half).
\label{eq:equil}
\end{equation}
The  number of  pure  points is  $|\partial\Sigma_{XZ}|=4$ whilst  the
dimension  is 2,  entailing nonsimpliciality.   Clearly, in  this case
$\Sigma \ne \Sigma_X \uotimes \Sigma_Z$. \hfill\qed
\end{example}

We stress that nonsimpliciality does not imply the non-existence of JD
in an underlying outcome-deterministic theory.  For instance, for each
state (\ref{eq:uneq}),  a JD manifestly  exists for $X$ and  $Z$.  For
example,  corresponding  to  the   state  $|\psi_X^+)$,  we  have  the
separable form $P(x,z|X,Z) = P(x|X)P(z|Z)$, where $P(x|X) = (1,0)$ and
$P(z|Z) = (\frac{1}{2}, \frac{1}{2})$.  And  so on for the other three
pure states.

Indeed, for $m=2$, one can always find a JD.  In Section \ref{sec:JD},
we shall find  that only incongruence with a  specific structure gives
rise to absence of JD.  In Section \ref{sec:conseq}, we will explore a
number  of nonclassical  consequences of  nonsimpliciality in  regular
theories.

\section{Consequences of nonsimpliciality\label{sec:conseq}}

If  a  regular  theory  $\mathcal{T}$ has  two  or  more  incongruent
measurements, then the space $\Sigma$  characterizing the theory is not
a simplex, as noted.  Interestingly, a number of nonclassical features
follow  from the  nonsimpliciality of  $\Sigma$, as  discussed in  the
following subsections.

\subsection{Co-measurability\label{sec:nonco}}

Two measurements  $X$ and  $Z$ each with  $n$ measurement  outcomes are
jointly  measurable  (or,  co-measurable)  if  there  exists  a  joint
measurement  $M_{XZ}$  such   that  the  statistics  of   $X$  and  $Z$
measurements are obtained under marginalization: $M_X(x|\psi) = \sum_z
M_{XZ}(x,z|\psi)$ and $M_Z(z|\psi) =  \sum_x M_{XZ}(x,z|\psi)$ for any
state $\psi  \in \Sigma_{XZ}$  \cite{BHS+0}.  In the  classical world,
all  measurements are  jointly  measurable trivially.  Co-measurability
roughly corresponds to  the commutativity of operators  in the Hilbert
space formulation of quantum mechanics.

Now consider  a nonclassical  theory with $m=2$  fiducial measurements,
denoted $X$  and $Z$.  They are taken  to be sharp  in the  sense that
there are states (namely, the  eigenstates) for which the measurements
produce definite  values. If  these two  are jointly  measurable, then
there   are  $(n^2-1)   \times  |\partial\Sigma|$   independent  terms
$M_{XZ}(x,z|\psi)$  that  constitute   the  putative  joint  operator,
considering  each  of  the  possible  $|\partial\Sigma|$  pure  states
$\psi$.   There  are   $2(n-1)  \times  |\partial\Sigma|$  independent
constraints  on them  due  to  the requirement  that  $M_X$ and  $M_Z$
statistics should be reproduced under marginalization.

A   nonsimplicality   condition,   which    is   of   the   type   Eq.
(\ref{eq:constraint}),   implies   additional  $(n^2-1)$   independent
constraints  of  the  type  $\sum_j q_j  M(x,z|\psi_j)  =  \sum_k  r_k
M(x,z|\psi_k)$,  which express  the  idea that  the joint  measurement
cannot be used to distinguish between  the mixtures in the LHS and RHS
of  (\ref{eq:constraint}).  These  constraints  overconstrain $M$  and
thus prohibit  joint measurability  only if  $2(n-1)|\partial\Sigma| +
n^2-1 > (n^2-1)|\partial\Sigma|$, or
\begin{equation}
|\partial\Sigma|\le\frac{n+1}{n-1} \le 3,
\label{eq:*constr}
\end{equation}
for  $n\ge2$.   Clearly,  this  can   never  happen  for  a  realistic
theory. Therefore, nonsimpliciality  by itself does not  imply lack of
joint  measurability.  However,  for  reasons explained,  we are  only
interested   in   regular   theories,    and   the   consequences   of
nonsimpliciality  in  them. In  this  case,  the relationship  between
simplciality of  $\Sigma$ and joint  measurability is tight,  as shown
below. For a different approach, see \cite{Pla16}.

\begin{thm}
Any pair of  measurements $X$ and $Z$ in a  regular theory are pairwise
co-measurable  iff their  associated  state space  $\Sigma_{XZ}$ is  a
simplex.
\label{thm:noJO}
\end{thm} 
\begin{proof}
``Only if''  direction: Let the  putative joint measurement  be denoted
  $M$.  Let the  ``eigenvalues'' of $X$ and $Z$ be  $x$ and $z$ taking
  values from 1  to $n$.  The following argument holds  for any one of
  the $2n$ these eigenstates.  Without  loss of generality, let $\psi$
  be  the eigenstate  of $X$  associated with  eigenvalue $x=1$.   The
  $M_X$ marginalization constraints require that:
\begin{align}
M(1,1|\psi) &+ M(1,2|\psi) + \cdots + M(1,n|\psi) = 1 \nonumber\\
M(2,1|\psi) &+ M(2,2|\psi) + \cdots + M(2,n|\psi) = 0 \nonumber\\
\vdots  \nonumber \\
M(n,1|\psi) &+ M(n,2|\psi) + \cdots + M(n,n|\psi) = 0.
\label{eq:jointmeza}
\end{align}
Note  that   the  signature  $(1,0,0,\cdots,0)$  is   imposed  by  the
assumption of pure state preparability in a regular theory.  The $M_Z$
marginalization constraints require that:
\begin{align}
M(1,1|\psi) &+ M(2,1|\psi) + \cdots + M(n,1|\psi) = p_1 \nonumber\\
M(1,2|\psi) &+ M(2,2|\psi) + \cdots + M(n,2|\psi) = p_2 \nonumber\\
\vdots \nonumber \\
M(1,n|\psi) &+ M(2,n|\psi) + \cdots + M(n,n|\psi) = \overline{p},
\label{eq:jointmezb}
\end{align}
where  $\overline{p}\equiv1-\sum_{j=1}^{n-1}p_j$.   This implies  that
all $M(j,k|\psi)$ vanish  for $j>1$ and the  $(n-1)$ independent terms
$p_j$  fix  the  terms  $M(1,k|\psi)$.   Thus,  clearly,  the  $(n-1)$
independent   numbers   $p_j$   completely   fix   the   elements   of
$M(j,k|\psi)$.   Repeating this  argument for  each of  the $2n$  pure
states generating  $\Sigma_{XZ}$, we  see that the  outcome statistics
for the  pure states fixes all  the elements of $M(j,k|\rho)$  for any
state, pure or mixed.

In general, the $p_j$'s are independent of each other.  Therefore, the
further $n^2-1$  independent constraints  due to  the nonsimpliciality
conditions  of  the  type  (\ref{eq:equil})  would  overdetermine  the
elements of $M$. Classical intuition  would suggest that no such extra
constraints are  possible.  But nonclassical phenomena  show that such
constraints can  exits and  essentially imply  that our  assumption of
existence of the joint measurement, is wrong.

The ``if''  direction: Intuitively, the challenge  to co-measurability
happens   inside  $\Sigma_{XZ}$   rather   than   outside.   Now,   if
$\Sigma_{XZ}$  is  a simplex,  then  the  above  operator $M$  can  be
constructed  uniquely with  regard to  $M$'s action  on the  states in
$\Sigma_{XZ}$. Consider pure states outside this space, of which there
are    $|\partial\Sigma|-|\partial\Sigma_{XZ}|$,    and   which    are
eigenstates of neither  $X$ nor $Z$. 

Assume for  simplicity that all  systems have $m+2$  inputs (including
$X$       and      $Z$)       and       $n$      outputs.        Then,
$|\partial\Sigma|-|\partial\Sigma_{XZ}|=mn$  and  there  are  $(n^2-1)
\times mn$ free variables, whereas the number of constraints are
\begin{subequations}
\begin{align}
c &= 2(n-1) \times mn + 
(n^2-1)\times\nu \label{eq:ish0} \\
&\le m(n-1)(3n+1)-(n^2-1),
\label{eq:ish1}
\end{align} 
\end{subequations}
where the first term  in the RHS of Eq. (\ref{eq:ish0})  is due to the
marginalizations and  $\nu\le m-1$  is the number  of nonsimpliciality
conditions of the type (\ref{eq:constraint})  in a regular theory.  On
the  number   hand,  the  number   of  free  variables  of   the  type
$M_{P,Q}(p,q|\psi)$ given pure state $\psi$, is
\begin{align}
v &= (n^2-1) \times mn \nonumber \\
  &= m(n-1)(n+1)n,
\label{eq:ish2}
\end{align}
Inasmuch as $c < v$ for $m\ge1$  and $n\ge2$, it follows that $M$ does
not get overconstrained.  
  \hfill \qed
\end{proof}

In regard to  the above proof, we  may also show a  similar result for
the case $m>2$,  by extending the above argument  in a straightforward
if elaborate way.

% This  is presented in  the Appendix \ref{sec:nonco+}.

Consider the 2-input-2-output case. Comeasurability  of $X$ and $Z$ in
the classical case is obvious.  We illustrate non-comeasurability in a
nonclassical theory.

\begin{example}
Consider  two measurements  $X$  and  $Z$, taking  values  $\pm1$, in  a
2-dimensional nonclassical theory, whose extreme states are:
\begin{subequations}
\begin{eqnarray}
|\psi_X^+) &\equiv& (1, 0~|~ 0.25, 0.75) \\
|\psi_X^-) &\equiv& (0, 1~|~ 0.75, 0.25) \\
|\psi_Z^+) &\equiv& (0.25, 0.75~|~ 1, 0)  \\
|\psi_Z^-) &\equiv& (0.75, 0.25~|~0, 1). 
\end{eqnarray}
\label{eq:list}
\end{subequations}
That the state space $\Sigma_{XZ}$ obtained  as the convex hull of the
state in (\ref{eq:list}) is not a simplex is seen by noting that
\begin{equation}
\frac{1}{2}(|\psi_X^+)  +   |\psi_X^-))  =   \frac{1}{2}(|\psi_Z^+)  +
|\psi_Z^-)).
\label{eq:psiX=Z}
\end{equation}
To see that  we cannot write down a joint  measurement $M_{XZ}$, assume
contrariwise that such $M_{XZ}$ exists.  To get the right marginalized
statistics for $|\psi_X^+)$, we require:
\begin{eqnarray}
M_{XZ}(+,+|\psi_X^+) + M_{XZ}(+,-|\psi_X^+) &=& 1 \nonumber \\
M_{XZ}(-,+|\psi_X^+) + M_{XZ}(-,-|\psi_X^+) &=& 0 \nonumber \\
M_{XZ}(+,+|\psi_X^+) + M_{XZ}(-,+|\psi_X^+) &=& 0.25 \nonumber \\
M_{XZ}(+,-|\psi_X^+) + M_{XZ}(-,-|\psi_X^+) &=& 0.75.
\end{eqnarray}
from    which   it    follows    that   only    $M_{XZ}(-,+|\psi_X^+)=
M_{XZ}(-,-|\psi_X^+)=0$   whereas  $M_{XZ}(+,+|\psi_X^+)=   0.25$  and
$M_{CS}(-,-|\psi_X^+)=0.75$.

Proceeding thus, one finds 
\begin{align}
&M_{XZ}(-,+|\psi_X^+) = 0; 
&M_{XZ}(-,+|\psi_Z^-) = 0
\nonumber\\
&M_{XZ}(-,+|\psi_X^-) = \frac{1}{4}; 
&M_{XZ}(-,+|\psi_Z^+) = \frac{3}{4}.
\label{eq:data}
\end{align}
To satisfy
the non-simpliciality condition, we must have
\begin{eqnarray}
&&M_{XZ}\left(j,k\left|\frac{1}{2}(|\psi_X^+) + 
|\psi_X^-))\right.\right)
= \nonumber\\
 &&~~ M_{XZ}\left(j,k\left|\frac{1}{2}(|\psi_Z^+) + 
|\psi_Z^-))\right.\right),
\end{eqnarray}
which  evidently   is  impossible,   as  seen  for   example,  setting
$(j,k)\equiv (-1,+1)$, since the LHS  yields $0.25/2$, whereas the RHS
yields $0.75/2$. \hfill \qed
\end{example}

It is important to stress that  lack of comeasurability does not imply
an absence of  JD in an underlying  outcome-deterministic theory.  For
example,  the  states  (\ref{eq:list})  have  manifestly  a  separable
(indeed,  product) form  $P(x,z|X,Z)  =  P(x|X)P(z|Z)$. 

\subsection{Measurement disturbance\label{sec:disturb}}

Measurement disturbance between two measurements refers to the possible
random  disturbance of  a  state  produced by  the  measurement of  an
measurement.  We  say that $X$  and $Z$  are mutually disturbing  if an
$X$-eigenstate  under  $Z$-measurement  yields  a  convex  mixture  of
$Z$-eigenstates and vice versa.

Intuitively,  the  nonsimpliciality  conditions  and  disturbance  are
closely  connected.   If  there  is no  measurement  disturbance,  all
measurements  can  be  performed  (infinitely many  times  if  required)
without   modifying  the   system,  and   all  pure   states  can   be
distinguished.  Thus, there would  be no indistinguishable mixtures of
the type $ \sum_{i=1}^m p_i\psi_i = \sum_{j=1}^n q_j\phi_j, $ and thus
no  nonsimpliciality  conditions. We  express  this  idea formally  as
following  theorem.    We  observe  that  without   the  nonsimplicial
conditions,  and   thereby  the  lack  of   corresponding  measurement
disturbance,  the   uncertainty  of  measuring  an   measurement  on  a
non-eigenstate would be  akin to the classical  probability, e.g., the
color of picked balls in a Polya urn \cite{KS06}.

\begin{thm}
Any pair of  measurements $X$ and $Z$ in a  regular theory are mutually
non-disturbing  iff their  associated state  space $\Sigma_{XZ}$  is a
simplex.
\label{thm:disturb}
\end{thm}
\begin{proof}  
Consider    regular   theory    $\mathcal{T}$   with    two   fiducial
noncomeasurable measurements, $X$ and $Z$.  Suppose that $X$ and $Z$ do
not  produce measurement  disturbance for  each other.   Then for  any
state in $\Sigma_{XZ}$,  we can perform $M_X$, and later  $M_Z$ on the
undisturbed state and  thereby construct the joint  measurement as per
the prescription
\begin{equation}
M_{XZ}|\psi)   \equiv    M_X(\psi)M_Z(\psi),
\label{eq:contradicting}
\end{equation}   
from   which  simpliciality   follows   per  Theorem   \ref{thm:noJO}.
Conversely,   if  $\Sigma_{XZ}$   is  a   simplex,  then   by  Theorem
\ref{thm:noJO},  measurements  $X$  and   $Z$  are  co-measurable.   In
particular, this means  that an unknown state in  $\Sigma_{XZ}$ can be
determined deterministically, making the joint measurement effectively
non-disturbing.  \hfill \qed
\end{proof}

It is not ruled  out that there may be states  in $\Sigma$ and outside
$\Sigma_{XZ}$ that  one or both may  disturb, even if $X$  and $Z$ are
mutually  non-disturbing.   The   following  example  illustrates  how
measurement disturbance enforces  the indistinguishability of mixtures
that are  operationally equivalent  by virtue of  the nonsimpliciality
conditions.

\begin{example} Suppose Alice prepares
(from  Bob's  perspective)  the   unbiased  $X$-mixture  $\rho  \equiv
  \frac{1}{2}(|\psi_X^+)    +   |\psi_X^-))$,    where   the    states
  $|\psi_X^\pm)$ are as defined  in Eq.  (\ref{eq:list}). This mixture
  is  operationally indistinguishable  from the  unbiased $Z$-mixture.
  Under measurement  of $X$,  the unbiased $X$-mixture  returns $\pm1$
  with equal  probability and leaves  the actual state, and  hence the
  mixture, undisturbed.

If  $Z$ is  measured, then  per (\ref{eq:list}),  both $Z$  values are
equiprobable, while the state after disturbance is:
\begin{equation}
\frac{1}{2}\left(\frac{1}{4}|\psi_Z^-) +
\frac{3}{4}|\psi_Z^+)\right) +
\frac{1}{2}\left(\frac{3}{4}|\psi_Z^-) +
\frac{1}{4}|\psi_Z^+)\right) = \rho,
\end{equation}
meaning that  the same mixture is  returned.  Any POVM in  this theory
(tossing  a loaded  coin and  measuring $X$  or $Z$  according to  the
coin's outcome) also does not  help in distinguishing the two mixtures
.\hfill \qed
\end{example}

In general,  the measurement disturbance entailed  by nonsimpliciality
is not unique.  Any  prescription for the post-measurement probability
distribution that preserves this indistinguishability is admissible.

\subsection{No-cloning\label{sec:nocloning}}

Cloning is an  operation by which, given an unknown  pure state $\psi$
in a state space, at least two  copies of $\psi$ are produced. We have
the following result.
\begin{thm}
An unknown pure state in the state space $\Sigma_{XZ}$ associated with
a pair of measurements $X$ and $Z$  in a regular theory is clonable iff
$\Sigma_{XZ}$ is a simplex.
\label{thm:noclon}
\end{thm}
\begin{proof}  
First,  we establish  the equivalence  of measurement  disturbance and
no-cloning.  Suppose in theory  $\mathcal{T}$, measurements $X$ and $Z$
do  not disturb  each other.   Then  each measurement  can be  measured
(repeatedly  if  required) without  disturbing  any  other.  From  the
resulting  classical record,  the state  of the  system is  completely
determined thus  cloned.  

For the other direction, we suppose that $\mathcal{T}$ permits perfect
cloning  of  an unknown  pure  state  $|\psi) \in  \Sigma_{XZ}$.   The
cloning operation can  be used to prepare multiple copies  of the same
state on any  number (say $t$) of other single  systems.  Now, suppose
that $X_1, X_2,  \cdots, X_t$ represent measurement trials  of a fixed
measurement (say  $X$) on  $t$ clones so  made.  Thus,  $X_j$ represent
independent and identically distributed random variables each of which
takes $n$ discrete values corresponding to the outcomes.

% Given fiducial  observables, $X, Z, \cdots$,  with the corresponding
% outcome  probabilities  being  $\vec{\mu}, \vec{\nu},  \cdots$,  the
% state of system is $\psi \equiv (\vec{\mu}~|~\vec{\nu}~|\cdots)$.

Let  the  $n$  frequencies  $f_j$   obtained  by  the  $t$  trials  be
represented by  the $n$-dimensional vector $\vec{f}$.   Let $\epsilon$
be a constant  in the range $[0,1]$.  We use  the notation $|\vec{a}|$
to  denote  the  vector  of component-wise  absolute  values  and  the
expression $\vec{a}  \ge \vec{b}$  to denote the  event that  for each
component $j$, $a_j \ge b_j$.

Since  the $X_j$'s  are independent  and identically  distributed, the
following bound can be shown to hold for the $t$ trials:
\begin{equation}
\textrm{Pr}\left(\left|\vec{f}       -      \vec{\mu}\right|       \ge
\epsilon\vec{\mu}\right)        \le       2\exp\left[\frac{-\epsilon^2
    t}{3n}\right],
\label{eq:chernoffsenior}
\end{equation}
which  implies that  for large  number  of trials  $t$, each  observed
frequency $f_j$  converges exponentially fast towards  the probability
$\mu_j$.

To  prove (\ref{eq:chernoffsenior}),  let $X  = \sum_{j=1}^t  X_j$ and
$\overline{X} \equiv  X/t$, where $X_j  \in [0,1]$.  By  the two-sided
Chernoff bound \cite{T09},
\begin{equation}
\textrm{Pr}\left(|\overline{X} - \mu| \ge \epsilon\mu\right) \le
2\exp\left[\frac{-\epsilon^2 \mu t}{2+\epsilon}\right],
\label{eq:chernoff}
\end{equation}
where $\mu$  is the theoretical mean.   Now divide $t$ into  $n$ equal
segments.  On the $k$th  segment, define the coarse-grained measurement
given by the  binary measurement $X^{(k)}$ which takes the  value 1 if
$X^{(k)}=x_k$, and  0 otherwise. Therefore the  $k$th segment involves
$t/n$  trials   of  random  variable  $X_j^{(k)}   \in  \{0,1\}$  with
probabilities $\{\mu_k,1-\mu_k\}$.

Applying the  Chernoff bound  (\ref{eq:chernoff}) to each  segment, we
can  bound $\textrm{Pr}(|\overline{X^{(k)}}-\mu_k|\ge  \epsilon\mu_k)$
by the right  hand side of (\ref{eq:chernoff})  with the substitutions
$t  \rightarrow t/n$  and $\mu\rightarrow  \mu_k$. The  conjunction of
these  $n$ events  is the  LHS of  (\ref{eq:chernoffsenior}), and  the
product    of    these    segmental    bounds   is    the    RHS    of
(\ref{eq:chernoffsenior}), where for simplicity,  we have replaced the
denominator   $2+\epsilon$    by   3   for   the    considered   range
$0\le\epsilon\le1$.

For a sufficiently large number of clones, we can perform a tomography
of either  measurement according  to Eq.   (\ref{eq:chernoffsenior}) by
measuring  sufficiently many  clones.   In this  way, the  probability
vector  associated   with  each   measurement  can  be   determined  to
reconstruct $|\psi)$ in $\Sigma_{XZ}$.

This establishes the equivalence  of disturbance and no-cloning, which
establishes the present theorem  in view of Theorem \ref{thm:disturb}.
\hfill \qed
\end{proof}
 
Refs.  \cite{MAG06,Bar07,BBL+07}  derive   a  no-cloning  theorem  for
general probabilistic theories more general  than QM in the context of
\textit{multipartite} systems subject to the no-signaling condition.

\subsection{Joint distinguishability\label{sec:jodi}}

States $\psi$ are jointly distinguishable if there is measurement that
uniquely  identifies the  individual states.   The \textit{measurement
  dimension}  $N$   is  the  number   of  pure  states  that   can  be
distinguished jointly,  i.e., by  a one-shot measurement  by measuring
one  or  more  fiducial   measurements  jointly.   Classical  intuition
suggests $D+1=N$, where the +1 in the LHS is due to normalization.  In
a  general  nonclassical  theory,   the  tomographic  and  measurement
dimensions do not match and we have:
\begin{equation}
N \le D+1,
\label{eq:tg}
\end{equation}
allowing  for the  possibility that  not all  pure states  are jointly
distinguishable.

\begin{thm}
All pure  states in  the state space  $\Sigma_{XZ}$ associated  with a
pair  of measurements  $X$  and $Z$  in a  regular  theory are  jointly
distinguishable iff $\Sigma_{XZ}$ is a simplex.
\label{thm:disting}
\end{thm}
\begin{proof} 
If $\Sigma_{XZ}$ is a simplex,  then by Theorem \ref{thm:disturb}, $X$
and $Z$ are mutually  nondisturbing.  Therefore, by sufficiently large
number of  repetitions (and  application of  the Chernoff  bound), any
given pure  state in  $\Sigma_{XZ}$ can  be uniquely  identified.  The
entire  repetitive   process  should  be  considered   as  a  one-shot
procedure.

To  prove the  converse:  let $\mathcal{P}$  denote  the procedure  to
jointly distinguish all  pure states in the  state space $\Sigma_{XZ}$
and  $\mathcal{Q}$ be  a  \textit{classical} cloner  that can  prepare
multiple  copies of  a  \textit{known} state  in $\Sigma_{XZ}$.  Then,
clearly  $\mathcal{Q} \circ  \mathcal{P}$ can  clone any  unknown pure
state in $\Sigma_{XZ}$, which per Theorem \ref{thm:noclon} entails the
simpliciality of $\Sigma_{XZ}$.  \hfill\qed
\end{proof}

Theorem  \ref{thm:disting}  implies  that  for  a  regular  theory,  a
simplicial space  $\Sigma$ is necessarily  the convex hull  of jointly
distinguishable  states, and  the  vertices of  a nonsimplicial  space
necessarily   lack   joint  distinguishability.    Given   regularity,
indistinguishability  has a  purely geometric  origin.  A  theory with
$\Sigma$ given as the simplex  of states that are \textit{not} jointly
distinguishable must  have some  other mechanism for  preventing joint
distinguishment, and would not be regular.

 The  preceding four  theorems  imply  an equivalence  among  the
nonclassical   properties   such   as   measurement   incompatibility,
measurement disturbance, etc., in the context of regular theories.  In
particular, Theorems  \ref{thm:noclon} and  \ref{thm:disting} together
imply  a  cloning-discrimination  equivalence,  essentially  using  an
approach  based  on  state  tomography. If  instead  one  uses  binary
observation tests between pairs of  states, then only three clones are
sufficient in each pairwise test to establish the simplex structure of
the state space \cite[Theorem 12]{CDP10}. 

\section{Uncertainty from nonsimpliciality in 
regular theories \label{sec:uncertain}}

Measurement  uncertainty refers  to the  property that  two (or  more)
measurements cannot simultaneously take definite values, as revealed by
a  measurement.   Our  definition  of a  nonclassical  regular  theory
presumes uncertainty.  In this  Section, we shall  physically motivate
that  assumption by  showing that  a theory  with nonsimplicial  state
space equipped with pure  state preparability will necessarily contain
measurement uncertainty.

\subsection{Uncertainty measure\label{sec:unc}}

It is convenient to quantify (measurement) uncertainty in a theory for
any two variables  $X$ and $Z$, with respective outcomes  $x$ and $z$,
as follows:
\begin{align}
\mathcal{U} =   \max_\psi\left[1-\left(\max_{x,z} 
\left[\frac{p(x|X,\psi) + p(z|Z,\psi)}{2}\right]\right)\right],
\label{eq:unc}
\end{align}
where $\psi$ runs over all states of the theory (cf. \cite{OW10}).  It
is  straightforward to  generalize  definition  (\ref{eq:unc}) to  $m$
$(\ge 2)$  measurements, essentially by  maximizing the average  of the
$m$ outcome probabilities over all $m$-outcome strings.

For the classical  theory, $\mathcal{U}=0$ for any  two properties $X$
and  $Z$.   For the  regular  nonclassical  theories characterized  by
states     (\ref{eq:list})     and    (\ref{eq:uneq}),     we     find
$\mathcal{U}=\frac{1}{8}$ and $\mathcal{U}=\frac{1}{4}$, respectively.

A theory  with uncertainty  need not  have nonsimplicial  state space.
This is because the space $\Sigma$  may be the convex hull of linearly
independent  vertices having  uncertainty in  the sense  of definition
(\ref{eq:unc}).  An instance  of such a theory is the  one whose state
space $\Sigma$ is the convex hull of four linearly independent states:
\begin{eqnarray}
|\psi_1) &\equiv& (1, 0~|~ 1,0) \nonumber \\
|\psi_2) &\equiv& (1, 0~|~ p, \overline{p}) \nonumber \\
|\psi_3) &\equiv& (p, \overline{p}~|~ 1, 0) \nonumber \\
|\psi_4) &\equiv& (p, \overline{p}~|~p, \overline{p}),
\label{eq:list-hash}
\end{eqnarray}
where  $\overline{p}=1-p$  and  $0\le  p\le1$.  Thus,  $\Sigma$  is  a
simplex.  Such a  theory has none of the no-go  theorems proven in the
previous section,  and hence the  uncertainty in such a  theory should
properly be regarded as due to classical ignorance. Such a ``Polya urn
theory''  is  really  a  classical theory  with  some  measurement  or
observation limitations.

\subsection{Gdit theories\label{sec:gdit}}

While   uncertainty  does   not  imply   nonsimpliciality,  conversely
nonsimpliciality  too does  not  entail  uncertainty.  A  \textit{gdit
  theory}  (gdit for  ``generalized dit'')  is one  with nonsimplicial
space   $\Sigma$   and  $\mathcal{U}=0$.    The   pure   state  in   a
$m$-input-$n$-output  gdit  theory  $\mathcal{T}_G$ is  one  of  $n^m$
\textit{boxes} represented  by an  $m$-ary vector $(x_1,  x_2, \cdots,
x_m)$,  where $x_j  \in \{1,\cdots,n\}$.   Measuring measurement  $X_j$
returns  $x_j$   deterministically.   An  arbitrary  mixed   state  is
described by $m$ probability distributions  with $n$ outcomes, so that
the tomographic dimension of the  states space $\Sigma_G$ is $m(n-1)$.
Therefore,  there are  exponentially many  relations between  the pure
states,  rendering them  linearly non-independent  and thereby  making
$\Sigma_G$ nonsimplicial. As proven below  later, a gdit theory is not
regular, because it lacks pure state preparability.

The simplest  gdit theory  $\mathcal{T}_G^{2,2}$ is  the 2-dimensional
theory, given  by the  deterministic version of  Eq.  (\ref{eq:uneq}).
The space $\Sigma_G^{2,2}$ is the convex hull of
\begin{eqnarray}
|g_0) &\equiv (0,0) &\equiv (1,0~|~1,0) \nonumber \\ 
|g_1) &\equiv (1,0) &\equiv (0,1~|~1,0) \nonumber \\
|g_2) &\equiv (0,1) &\equiv (1,0~|~0,1) \nonumber \\
|g_3) &\equiv (1,1) &\equiv (0,1~|~0,1).
\label{eq:ceq}
\end{eqnarray}
This satisfies the condition:
\begin{align}
\half(|g_0) + |g_3)) &= \half(|g_1) + |g_2)) \label{eq:03|12}
\\
&= (\half,\half), \nonumber
\end{align}
which, in our framework, means that the mixtures on the LHS and RHS of
(\ref{eq:03|12}) are indistinguishable.

\subsection{Symmetric and asymmetric gdit theories}

In a \textit{symmetric gdit}  theory, the prescription for measurement
disturbance is that when one of  the $m$ measurements is performed, then
the definite value of the performed measurement remains the same but the
value  assignments  to  the   unmeasured  $m-1$  measurements  will  be
equi-probably distributed:
\begin{align}
(x_1,x_2,\cdots, x_i,&\cdots,x_m)
    \stackrel{X_i}{\longrightarrow}
d^{1-m} \times \nonumber\\
 & \sum_{x_j^\prime \ne x_i} (x_1^\prime,x_2^\prime,\cdots, x_i,\cdots,x_m^\prime)
\end{align}
This  corresponds to  the  concept  of a  mutually  unbiased basis  in
quantum mechanics. For example, for the states in (\ref{eq:ceq})
\begin{align}
\left. \begin{array}{c}
|g_0) \\ |g_2)
\end{array}
\right\} \stackrel{X}{\longrightarrow}
\frac{1}{2}(|g_0) + |g_2)),\nonumber\\
\left. \begin{array}{c}
|g_1) \\ |g_3)
\end{array}
\right\} \stackrel{X}{\longrightarrow}
\frac{1}{2}(|g_1) + |g_3)),
\label{eq:post2}
\end{align}
whereby  the  $X$   value  is  unaltered  whilst  the   $Z$  value  is
irreversibly  replaced by  a uniform  distribution.  This  measurement
disturbance in the pattern of the measurement \textit{uncertainty} for
states  $\psi_X^\pm$ in  Eq.  (\ref{eq:uneq}). 

In   an  \textit{asymmetric   gdit}  theory,   the  prescription   for
measurement disturbance  is similar, except that  the post-measurement
distribution  of   values  of   the  unperformed  measurements   is  not
necessarily uniform:
\begin{align}
(x_1,x_2,\cdots, x_i,&\cdots,x_m)
    \stackrel{X_i}{\longrightarrow}
\sum_{x_j^\prime, j \ne i} p_{x_1^\prime,\cdots, 
x^\prime_{i-1},x^\prime_{i+1},\cdots,x_m^\prime}
 \nonumber\\
 & 
\times (x_1^\prime,x_2^\prime,\cdots, x_i,\cdots,x_m^\prime),
\label{eq:asgdit}
\end{align}
where  $\sum_{x_j^\prime,  j \ne  i}  p_{x_1^\prime,x_2^\prime,\cdots,
  x^\prime_{i-1},x^\prime_{x+1},\cdots,x_m^\prime}=1$.            Such
asymmetric disturbance  can be shown  to leave equivalent  mixtures of
pure states indistinguishable, as in the symmetric gdit theory.

For  example,   for  the   states  in  (\ref{eq:ceq})   an  asymmetric
prescription for disturbance could be:
\begin{align}
\left. \begin{array}{c}
|g_0) \\ |g_2)
\end{array}
\right\} \stackrel{X}{\longrightarrow}
\frac{1}{4}|g_0) + \frac{3}{4}|g_2),\nonumber\\
\left. \begin{array}{c}
|g_1) \\ |g_3)
\end{array}
\right\} \stackrel{X}{\longrightarrow}
\frac{1}{4}|g_3) + \frac{3}{4}|g_1),
\label{eq:post3}
\end{align}
which replaces  the symmetric  disturbance in  (\ref{eq:post2}).  Here
the  $X$ value  is unaltered  $X=\pm1$, respectively)  whilst the  $Z$
value is  distributed in the  pattern of the  \textit{uncertainty} for
states $|\psi_X^\pm)$  in Eq.  (\ref{eq:list}).  The  following result
shows that  the nonsimpliciality  conditions do  not uniquely  fix the
disturbance.

Generalizing  the  above  argument,   in  a  symmetric  or  asymmetric
$m$-input-$n$-output gdit theory,  the measurement disturbance ensures
that   the  indistinguishability   implied  by   the  nonsimpliciality
conditions   is   preserved.    To   see  this,   suppose   that   the
nonsimpliciality condition in the gdit theory has the form:
\begin{align}
\half[(0,0&,x_3,x_4,\cdots) +
(1,1,x_3,x_4,\cdots)] \nonumber \\
   &= \half[(0,1,x_3,x_4,\cdots) +
(1,0,x_3,x_4,\cdots)],
\label{eq:noex}
\end{align}
where $x_j$ $(j=1,2,\cdots)$ are outcomes  that would be obtained when
measurement $X_j$ is  performed.  Given the two  mixtures represented by
the LHS  and RHS of  (\ref{eq:noex}), suppose $X_1$ is  measured. Then
the post-measurement mixture due to  the first (resp., second) term in
the LHS  and the  first (resp., second)  term in the  RHS will  be the
same.  Thus, the post-measurement mixture is the same for both initial
states. Now, if  $X_2$ is measured, then  the post-measurement mixture
due  to the  first (resp.,  second)  term in  the LHS  and the  second
(resp., first)  term in the  RHS will  be the same.   Clearly, similar
arguments can be  given for all other $X_j$'s.   It is straightforward
to generalize the  above argument from bits to dits,  and to relations
between pure states more elaborate than Eq. (\ref{eq:noex}).

An  insight that  emerges from  noting the  joint measurability  in an
arbitrary nonclassical theory vs  the non-comeasurability in a regular
nonclassical theory (Section \ref{sec:nonco})  is that the idempotency
of the marginal distribution helps  ``kill'' degrees of freedom in the
joint measurement, facilitating  non-comeasurability.  By that yardstick,
a gdit theory, in which all pure states marginals are idempotent,
should readily  lack joint measurability. This is indeed
the case, as the following example shows.

\begin{example}
Consider  the 2-input-$n$-output  gdit theory.   Proceeding along  the
lines of  Eq.  (\ref{eq:jointmeza}), it  is readily observed  that the
joint measurement should,  if it exists, have the  form $M(j,k|(r,s)) =
\delta^j_r\delta^k_s$.   Now a  particular nonsimpliciality  condition
is:
\begin{equation}
\half((j,k) + (j^\prime,k^\prime)) = \half((j,k^\prime) + (j^\prime,k)),
\label{eq:simplicdit}
\end{equation}
where  $j  \ne  j^\prime$  and   $k  \ne  k^\prime$.  For  any  state,
Eq.    (\ref{eq:simplicdit})   constrains    the   joint    measurement
$M(a,b|\psi)$ such that:
\begin{equation}
M(j,k|\psi) + M(j^\prime,k^\prime|\psi) = M(j,k^\prime|\psi) + 
M(j^\prime,k|\psi).
\label{eq:simplicdit0}
\end{equation}
Suppose such a joint operator $M(a,b|\psi)$ exists.  Let $\psi=(j,k)$.
In Eq.   (\ref{eq:simplicdit0}), of  the four terms,  the first  is 1,
while  the remaining  vanish,  implying that  the  equation cannot  be
satisfied.   It  is easy  to  extend  this  argument to  an  arbitrary
$m$-input-$n$-output  gdit  theory.   Therefore, no  joint  measurement
exists for gdit theory.  \hfill \qed
\end{example}

As a canonical irregular theory, no sharp measurement in the
theory has an associated eigenstate.  
Thus, gdit  theory lacks  preparability.   For  any  state, there  is  no
measurement that leaves  it undisturbed. Therefore, a pure  state in a
gdit theory can  never be prepared, and state  preparations are always
ambiguous.  It turns
out that the preparation ambiguity in  a gdit theory can be matched to
the uncertainty of a regular theory, as discussed below.

\subsection{Uncertainty-disturbance correspondence\label{sec:corr}}

An   $m$-input-$n$-output  (symmetric   or  asymmetric)   gdit  theory
$\mathcal{T}_G$,   characterized   by    a   disturbance   probability
distribution   $p_{x_1^\prime,  x_2^\prime,   \cdots,  x_{i-1}^\prime,
  x_{i+1}^\prime  ,\cdots,  x_m^\prime}$  given  in  (\ref{eq:asgdit})
under  measurement  of  $X_i$  is said  to  \textit{correspond}  to  a
$m$-input-$n$-output regular theory $\mathcal{T}_U$, if the pure state
with definite value $X_i=x_i$ is given by
\begin{align}
|X_i=x_i)_U &= 
\sum_{x_j^\prime \ne x_i} p_{x_1^\prime,,\cdots, 
x_{i-1}^\prime,x_{i+1}^\prime,\cdots,x_m^\prime} \nonumber \\
  & \times 
(x_1^\prime,x_2^\prime,\cdots, x_i,\cdots,x_m^\prime)_G,
\end{align}
where the paranthesized term with subscript $G$ is a gdit state.  This
point was already  noted in connection with  Eqs. (\ref{eq:post2}) and
(\ref{eq:post3}).

The number of pure states in  theory $\mathcal{T}_G$ is $n^m$ while it
is $mn$ in  the corresponding regular theory  $\mathcal{T}_U$.  On the
other hand, both have the  same dimension $m(n-1)$, which implies that
$\mathcal{T}_G$ will  have exponentially more number  of nonsimplicial
conditions.    To  determine   the   nonsimpliciality  conditions   in
$\mathcal{T}_U$ starting from those in $\mathcal{T}_G$, we measure the
LHS  and  RHS  of  a  gdit nonsimpliciality  conditions  in  the  same
measurement,  and  re-interpret   the  resulting  measurement-disturbed
states as pure states with  uncertainty in the regular theory obtained
using the above $\mathcal{T}_G-\mathcal{T}_U$ correspondence.

For example,  to get from  Eq. (\ref{eq:03|12}) to  (\ref{eq:psiX=Z}), we
measure the LHS of (\ref{eq:03|12}) by $X$, which effects according to
the above recipe:
$$
|g_0) \rightarrow |\psi_X^+);\qquad
|g_3) \rightarrow |\psi_X^-),
$$
and we measure  RHS of (\ref{eq:03|12}) by $Z$, which effects:
$$
|g_1) \rightarrow |\psi_Z^+);\qquad
|g_2) \rightarrow |\psi_Z^-),
$$  from  which  (\ref{eq:psiX=Z}).  The gdit  and  the  corresponding
regular theories  have the  same measurements,  though the  states that
lack measurement uncertainty are different.

Interestingly, a gdit theory and  its corresponding regular theory are
\textit{operationally indistinguishable} in our framework, essentially
because of the interplay of  disturbance and uncertainty, as discussed
below.

For our purpose,  a theory with tomographic separability  can be fully
characterized  at  the  operational  level in  two  steps:  (1)  State
preparation: measuring  fiducial measurement $A$ and  post-selecting on
outcome $a$; (2) Measurement: determining the outcome probabilities of
a subsequent  measurement $B \ne  A$ of other measurements.   If $B=A$,
then the outcome is always assumed $a$ (i.e., repeatability assumed).

Suppose we have two  theories $\mathcal{T}_1$ and $\mathcal{T}_2$, and
two $n$-dimensional  measurements $X$  and $Z$ in  $\mathcal{T}_1$, and
their analogues in  $\mathcal{T}_2$, also referred to as  $X$ and $Z$.
In  either theory,  we can  prepare a  state by  measuring one  of the
measurements (say $X$) and post-selecting  on a particular outcome.  If
a subsequent measurement of the  other measurement (say $Z$) yields the
same outcome probabilities, and this  holds true for every preparation
in $\mathcal{T}_1$  and its analogous preparation  in $\mathcal{T}_2$,
then   $\mathcal{T}_1$    and   $\mathcal{T}_2$    are   operationally
indistinguishable.    An    observer    cannot    determine    whether
$\mathcal{T}_1$ or $\mathcal{T}_2$ is applicable in the operational world.

Note  that in  a gdit  theory, because  of preparation  ambiguity, the
first  step above  does not  produce an  unambiguous pure  state.  For
example, suppose that in the 2-input-2-output gdit theory, Alice needs
to prepare $(X=0,Z=0)$.   If she measures $X$ on an  unknown state and
postselects on 0, $Z$ remains indeterminate.  A subsequent measurement
on  $Z$  disturbs  $X$.   This ambiguity  in  preparation  mimics  the
measurement uncertainty  of the  corresponding regular theory,  as the
following result shows.

\begin{thm}
A  gdit theory  $\mathcal{T}_G$ and  its corresponding  regular theory
$\mathcal{T}_U$ are operationally indistinguishable from each other.
\label{thm:UG}
\end{thm}
\begin{proof}
Let us consider  a gdit theory $\mathcal{T}_G$ of  two measurements $X$
and  $Z$ (with  straightforward generalization  to more  measurements).
Without loss of generality, suppose  we first prepare $X=x$ during the
state preparation  step.  In the  gdit theory, if the  state preceding
measurement was such that $X\ne x$, then it will be ``post-rejected'',
i.e.,  post-selected out.   Thus,  a  post-selected state  necessarily
pre-possessed  definite value  $X=x$.  Therefore, no  matter what  the
unknown pure state preceding  the measurement, the post-selected state
under  the  subsequent  measurement  of  $Z$  will  reflect  the  same
disturbance  probability $p_X(z)$,  which will  be $\frac{1}{d}$  in a
symmetric gdit theory.

By  definition of  the corresponding  regular theory,  the preparation
step  above  will produce  the  unique  regular state  $|X=x)$,  whose
uncertainty  probability   distribution  will  be  $p_X(z)$.    A  $Z$
measurement in step (2) will thus  give the outcome statistics for the
gdit theory and its corresponding regular theory. \hfill \qed
\end{proof}

The ability of observers to prepare any pure state is obviously a nice
feature and hence an important requirement for any operational theory.
Requiring  preparability is  the  reason that  a nonclassical  regular
theory  has   $\mathcal{U}>0$.   This   justifies  our   inclusion  of
uncertainty   as   a   basic   feature  of   regularity   in   Section
\ref{sec:frame}.

The   uncertainty-disturbance   shows    that   the   regular   theory
corresponding  to  a  gdit  theory  can be  considered  as  having  an
\textit{epistemic interpretation}  over ontic  states supplied  by the
gdit theory.  This observation generalizes  to Spekkens' model  for an
epistemic model \cite{Spe07}, discussed later below.

In   Figure    \ref{fig:JD},   the    outer   square    represents   a
two-input-two-output gdit  theory, with  the two  fiducial measurements
being  $X$  and   $Z$.   The  inner  square  is   the  regular  theory
corresponding  to  the  symmetric  gdit  theory.   The  circular  disk
represents a  more general regular  theory, with infinitely  many pure
states  (circular rim),  generated by  the gdit  theory considering  a
continuum  of  asymmetric  measurement  disturbances.   The  antipodal
points of  the circle will  represent eigenstates of  rotated fiducial
measurements   $A$  and   $B$,  obtained   under  suitable   continuous
transformation of $X$ and $Z$.
\begin{figure}
\includegraphics[width=6cm]{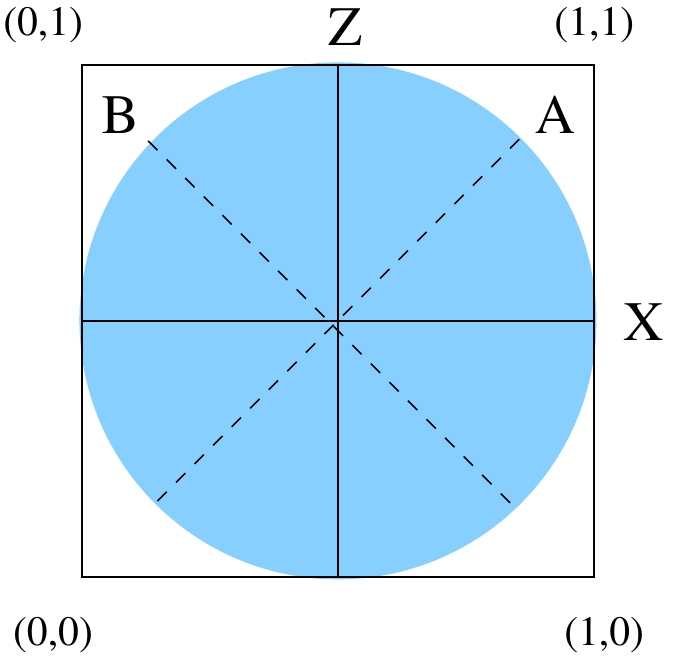}
\caption{The state space  of the two-input-two-put gdit  theory is the
  outer square. The inner square (resp., circular disk) is the regular
  theory  corresponding to  the  symmetric  (resp., assymmetric)  gdit
  theory.  The pairs $(X,Z)$ and $(A,B)$ represent equivalent fiducial
  measurements,  connected   by  a   continuous  rotation.    From  the
  perspective of $\Sigma$-ontology  (cf.  Section \ref{sec:prep}), the
  gdit   theory  is   intermediate  ontological   space  $\Sigma_\cup$
  underlying the  operational (regular) theory $\Sigma$,  which can be
  the inner square or the circle.}
\label{fig:JD}
\end{figure}

\section{Ontology of noncontextual nonclassical theories \label{sec:simplionto}}

The  preceding consequences  of  nonsimpliciality  follow from  purely
operational  considerations.   In this  Section,  we  take up  certain
ontological issues. Here our  approach is comparable with that of
\cite{CSW14,abramsky2011sheaf,silva2017graph},  which   explore  graph
theoretic  approaches to  nonclassicality  in  quantum mechanics.   We
return to this matter in Section \ref{sec:cbox}, but here we only note
that  in  contrast  to  the  three  aforementioned  works,  which  are
applicable to contextuality and nonlocality, our combinatoric approach
(as  explored in  this Section)  works also  for nonclassical  effects
arising  in   the  absence  of  contextuality   and  nonlocality,  but
associated with the nonsimpliciality of the state space.

\subsection{Underlying simplex and $\Sigma$-ontology
\label{sec:sqcup}}

The idea that an ontological model for a nonclassical theory is itself
a  classical  theory  of  some  kind  suggests  an  \textit{underlying
  simplex}  $\Sigma_\sqcup$  and  a  corresponding  \textit{underlying
  classical theory} $\mathcal{T}_\sqcup$ as natural representations of
an ontological model. Given fiducial  measurements $X, Y, Z, \cdots$ in
$\mathcal{T}$,  we posit  underlying  versions  of these  measurements,
which will also be referred to  as $X, Y, Z, \cdots$.  However, unlike
in $\mathcal{T}$, these  measurements \textit{do} in the underlying
theory peacefully
coexist.  Therefore, $\Sigma_\sqcup  \equiv \Sigma_X \uotimes \Sigma_Z
\uotimes  \cdots$.    Since  the  individual  spaces   $\Sigma_j$  are
simplexes, so is $\Sigma_\sqcup$.  So  long as an underlying JD exists
for  the properties,  such an  underlying theory  $\mathcal{T}_\sqcup$
exists. The  ontological model that  it can provide for  the overlying
operational theory  $\mathcal{T}$ is discussed below.   In particular,
the  elements  $\partial\Sigma_\sqcup$  are taken  to  constitute  the
underlying ontic states $\lambda$.

The  relationship   between  the   simplex  $\Sigma_\sqcup$   and  the
nonsimplex  $\Sigma$  of  $\mathcal{T}$, represented  by  the  mapping
$\varphi:    \Sigma_\sqcup    \rightarrow     \Sigma$,    cannot    be
linear-injective, since this would  preserve the simplex property.  It
is convenient  to think of $\varphi$  as the composition of  two maps:
$\varphi \equiv \varphi_\star \circ \varphi_\triangledown$.

The    map     $\varphi_\triangledown:    \Sigma_\sqcup    \rightarrow
\Sigma_\cup$, where $\Sigma_\cup$ is either an intermediate underlying
simplex or an intermediate underlying gdit space such that
\begin{equation}
\textrm{dim}(\Sigma_\cup) = \textrm{dim}(\Sigma).
\label{eq:cupeq}
\end{equation}
Therefore, for a nonclassical $m$-input-$n$-output regular theory with
pairwise incongruent measurements,  $\varphi_\triangledown$ reduces the
dimensionality exponentially  from $\textrm{dim}(\Sigma_\sqcup)=n^m-1$
to $\textrm{dim}(\Sigma)=m(n-1)$.   It is  convenient to call  the map
$\varphi_\triangledown$ as \textit{compression}.

Typically, map $\varphi_\triangledown$ does not introduce uncertainty.
This   introduction    happens   via    the   partial    function   of
\textit{crumpling},  $\varphi_\star: \Sigma_\cup  \rightarrow \Sigma$,
which is such that $\varphi_\star^{-1}:\Sigma \rightarrow \Sigma_\cup$
is an embedding.  We refer  to $\varphi_\star$ as a \textit{crumpling}
map  because  it introduces  vertices  and  faces.   We call  the  map
$\varphi$ and  the corresponding  ontological model as  g-type (resp.,
s-type) if space $\Sigma_\cup$ is a gdit space (resp., simplex).

The  term  \textit{$\Sigma$-ontology}  will   refer  to  this  general
procedure for  constructing an ontological  model on the basis  of the
underlying  simplex  $\Sigma_\sqcup$  and the  intermediate  underlier
$\Sigma_\cup$. An example  is given by Figure  \ref{fig:JD}, where the
regular theory is the inner square. The the outer square represents an
intermediate  underlying theory  $\Sigma_\cup$, while  the tetrahedron
(not shown) is the underlying simplex $\Sigma_\sqcup$.

\subsubsection{G-type ontological model}

Here the intermediate space  $\Sigma_\cup$ is the $m$-input-$n$-output
gdit theory, which one may  call the \textit{underlying gdit theory}.
It preserves  the $n^m$  extreme points  of $\Sigma_\sqcup$  such that
$\varphi_\triangledown:          \partial\Sigma_\sqcup         \mapsto
\partial\Sigma_\cup$ is a one-to-one-correspondence.  If the points in
$\Sigma_\sqcup$  are represented  in a  barycentric coordinate  system
$\mathcal{B}$, then  those in $\Sigma_\cup$  can be expressed  using a
generalized barycentric  coordinate system $\mathcal{G}$  derived from
$\mathcal{B}$   using    $\varphi_\triangledown$.    Owing    to   the
nonsimpliciality   of  the   gdit   theory,   the  representation   in
$\mathcal{G}$ in general won't be unique.

\begin{example}
Suppose there are two fiducial measurements  $X, Z$ taking values $x, z
\in  \{0,1\}$ in  the  operational theory  characterized  by the  pure
states  in   Eq.   (\ref{eq:uneq}).    Here  the   underlying  simplex
$\Sigma_\sqcup$ is the 3-simplex generated by four vertices $\lambda_j
\equiv  x  \otimes  z$  ($j \in  \{1,2,3,4\}$).   In  the  barycentric
coordinate   system,  the   $\lambda_j$'s  are   the  states   in  Eq.
(\ref{eq:ontixtate}).

Employing the  generalized barycentric coordinates  for $\Sigma_\cup$,
we can express the action of $\varphi_\triangledown$ on pure states of
$\Sigma_\sqcup$  by $\varphi_\triangledown(x  \otimes z)  = (x,z)  \in
\mathbb{R}^2$.  Given  a point  in  $\Sigma_\sqcup$  described by  the
barycentric coordinates $(a_1,a_2,a_3,a_4)$,  where $a_j\ge0$, $\sum_j
a_j=1$ , we have
\begin{align}
\varphi_\triangledown[(a_1&,a_2,a_3,a_4)] = 
\sum_{j=1}^4 a_j \varphi_\triangledown(\lambda_j)
\nonumber\\
&= a_1(0,0) + a_2(0,1)+a_3(1,0)+a_4(1,1) \nonumber\\
&= (a_3+a_4,a_2+a_4),\nonumber\\
&\equiv (a_1+a_2,a_3+a_4 \mid a_1+a_3,a_2+a_4)
\label{eq:GBC}
\end{align}
where  $(0,0), (0,1),  (1,0)$  and  $(1,1)$ are  the  vertices of  the
underlying  gdit  theory.   Note  that  distinct  interior  points  in
$\Sigma_\sqcup$ may map to the  same point in $\Sigma_\cup$, since the
vertices    $\varphi_\triangledown(\lambda_j)$   are    not   linearly
independent.     For   instance,    the   distinct    ontic   mixtures
$(\half,0,0,\half)$ and $(0,\half,\half,0)$  in the underlying simplex
map  to the  same  mixture $(\half,\half)  \equiv  (\half, \half  \mid
\half, \half)$ in the gdit theory. \hfill \qed
\end{example}

In  Figure \ref{fig:JD},  space  $\Sigma_\cup$ is  represented by  the
outer square and operational space $\Sigma$ by the circle or the inner
square.   Map  $\varphi_\star^{-1}$ embeds  the  circle  in the  outer
square  $\Sigma_\cup$.   The square  is  obtained  by compressing  set
$\Sigma_\sqcup$ (a  3-simplex, not  shown in the  figure) via  the map
$\varphi_\triangledown$.   This  compresses  the  tetrahedron  to  the
square.

\begin{example}
In our approach,  Spekkens' toy theory \cite{Spe07}  has the following
operational  description.    There  are  three   pairwise  incongruent
measurements $X, Y, Z$, with $\Sigma$ given as the convex hull of:
\begin{align}
\psi_X^+ &\equiv (1,0\mid \half, \half \mid \half, \half) \nonumber\\
\psi_X^- &\equiv (0,1\mid \half, \half \mid \half, \half) \nonumber\\
\psi_Y^+ &\equiv (\half, \half \mid 1,0 \mid \half, \half) \nonumber\\
\psi_Y^- &\equiv (\half, \half \mid 0, 1 \mid \half, \half) \nonumber\\
\psi_Z^+ &\equiv (\half, \half \mid \half, \half \mid 1,0) \nonumber\\
\psi_Z^+ &\equiv (\half, \half \mid \half, \half \mid 0,1).
\label{eq:spekkens}
\end{align}
The nonsimpliciality of $\Sigma$ is a consequence of dependences among
the extreme states:
\begin{align}
 (\half,\half\mid \half,\half \mid \half,
\half) &= \half(\psi_X^+   +   \psi_X^-) \nonumber\\
 &=  \half(\psi_Y^+   +   \psi_Y^-)\nonumber\\
   &=
\half(\psi_Z^+ + \psi_Z^-).
\label{eq:speksimplex}
\end{align}
Therefore  the various  nonclassical features  noted in  Spekkens' toy
theory are  seen to  follow essentially  from the  nonsimpliciality of
$\Sigma$ in that theory.

The  underlying simplex is
\begin{equation}
\Sigma_\sqcup = \Sigma_X \uotimes  \Sigma_Y
\uotimes \Sigma_Z,
\label{eq:mps}
\end{equation}
of dimensionality $\textrm{dim}(\Sigma_\sqcup)=7$.   Since for the toy
theory  $D=3$,   therefore  $\Sigma_\cup$  is   the  three-dimensional
3-input-2-output  gdit theory,  with vertices  given by  the 8  states
$\gamma_{4^j+2^k+l+1}  \equiv (j,k,l)$  in Cartesian  coordinates, and
$j,k,l   \in    \{0,1\}$.    The    compression   map   is    given   by
$\varphi_\triangledown(j \otimes  k \otimes  l) = (j,k,l)$,  where the
ontic elements  $\lambda \equiv  (j \otimes  k \otimes  l)$ are  the 8
extreme  points of  the 7-simplex  $\Sigma_\sqcup$, and  distinct from
Spekkens' ontic set  of 4 points (which are the  extreme points of the
$\Sigma_\cup$ in the s-type ontology, discussed below).

In the  g-type ontology,  the Spekkens toy  theory corresponds  to the
octahedron formed  as the convex hull  of the six face-centers  of the
3-input-2-output gdit theory.  The  ontic representation of the states
(\ref{eq:spekkens}) in terms of the intermediate gdit theory is:
\begin{align}
\mu_X^+ &= \frac{1}{4}(\gamma_1+\gamma_2+\gamma_3+\gamma_4) \nonumber\\
\mu_X^- &= \frac{1}{4}(\gamma_5+\gamma_6+\gamma_7+\gamma_8) \nonumber\\
\mu_Y^+ &= \frac{1}{4}(\gamma_1+\gamma_2+\gamma_5+\gamma_6) \nonumber\\
\mu_Y^- &= \frac{1}{4}(\gamma_3+\gamma_4+\gamma_7+\gamma_8) \nonumber\\
\mu_Z^+ &= \frac{1}{4}(\gamma_1+\gamma_3+\gamma_5+\gamma_7) \nonumber\\
\mu_Z^- &= \frac{1}{4}(\gamma_2+\gamma_4+\gamma_6+\gamma_8),
\label{eq:spekkensmu}
\end{align}
which are  just the  generalized barycentric  coordinates in  the gdit
theory.
\end{example}
A geometrization  of Spekkens' toy  theory using precisely  four ontic
points, based on s-type $\Sigma$-ontology, is discussed below.

\subsubsection{S-type ontological model}

In this type  of model, the intermediate space  $\Sigma_\cup$ is taken
to  be the  $D$-simplex,  where $D\equiv\textrm{dim}(\Sigma)$.   Thus,
$\varphi_\triangledown$ is  a function  that preserves  the simplicial
property.     The    crumpling    map    $\varphi_\star$    introduces
nonsimpliciality,  in  addition  to  uncertainty.  This  idea  can  be
illustrated  by  giving  an  s-type model  for  Spekkens'  toy  theory
\cite{Spe07}, in contrast to the g-type given above.

\begin{example}
For the Spekkens toy theory, whose  $\Sigma$ is the convex hull of the
states  (\ref{eq:spekkens}),   the  intermediate   underlying  simplex
$\Sigma_\cup$  is  the  3-simplex,   or  tetrahedron.   The  compression
function $\varphi_\triangledown$  maps the  7-simplex $\Sigma_\sqcup$,
defined in (\ref{eq:mps}), to this.  Here $\Sigma_\cup$ is taken to be
the   convex   hull   of    any   four   convex   independent   points
$\gamma_j^\prime$  in  $\mathbb{R}^3$.   In particular,  we  can  take
$\gamma_1^\prime  \equiv  (0,0,0),   \gamma_2^\prime  \equiv  (0,1,1),
\gamma_3^\prime \equiv (1,0,1)$ and $\gamma_4^\prime \equiv (1,1,0)$.

The  state space  $\Sigma$ in  Spekkens toy  theory is  the octahedron
formed  as the  convex  hull  of the  six  face-centers  of the  above
3-simplex,    which   serves    as   $\Sigma_\cup$,    with   vertices
$\{\gamma^\prime_j\}$  ($j\in\{1,2,3,4\}$).  The  ontic representation
in terms of the intermediate underlying simplex is:
\begin{align}
\mu_X^{+\prime} &= \frac{1}{2}(\gamma_1^\prime+\gamma_2^\prime) \nonumber\\
\mu_X^{-\prime} &= \frac{1}{2}(\gamma_3^\prime+\gamma_4^\prime) \nonumber\\
\mu_Y^{+\prime} &= \frac{1}{2}(\gamma_1^\prime+\gamma_3^\prime) \nonumber\\
\mu_Y^{-\prime} &= \frac{1}{2}(\gamma_2^\prime+\gamma_4^\prime) \nonumber\\
\mu_Z^{+\prime} &= \frac{1}{2}(\gamma_1^\prime+\gamma_4^\prime) \nonumber\\
\mu_Z^{-\prime} &= \frac{1}{4}(\gamma_2^\prime+\gamma_3^\prime),
\label{eq:spekkensmu0}
\end{align}
which   may   be   contrasted    with   the   g-type   representation,
Eq. (\ref{eq:spekkensmu}). 
\hfill \qed
\end{example}
This  geometric  containment of  state  space  $\Sigma$ in  the  above
3-simplex for Spekkens' toy theory's was first shown in \cite{DGV14}.

\subsection{Impossible universal operations\label{sec:impossible}}
 
We  define a  \textit{coherent transformation}  $U$ in  an operational
$\mathcal{T}$ as  one that reversibly  maps any pure state  to another
pure state.   An instance of  a coherent  operation for a  theory with
$N=2$ is the inverter, which transforms  any state to another which is
orthogonal  in the  sense of  being deterministically  distinguishable
from the initial state. Thus, the inverter is fully specified by
\begin{eqnarray}
|\psi_X^+) &\longleftrightarrow& |\psi_X^-), \nonumber\\
|\psi_Z^+) &\longleftrightarrow& |\psi_Z^-),
\label{eq:invertop}
\end{eqnarray}
for the theory determined by  pure states (\ref{eq:uneq}).  In quantum
mechanics,  it  is  known  that there  exists  no  universal  inverter
\cite{RCC+01}.   Spekkens   showed  \cite{Spe07}  that   an  analogous
non-invertibility  result  exists  is  his model.   His  proof,  which
employs an ontological argument, proceeds along the following lines.

In the  ontic  representation  (\ref{eq:spekkensmu0}) of  the
states (\ref{eq:spekkens}), to implement $\psi^+_X \longleftrightarrow
\psi_X^-$ and $\psi^+_Y \longleftrightarrow \psi_Y^-$, it is seen that
the requisite ontic transformation is
\begin{align}
\gamma_1^\prime  &\longleftrightarrow   \gamma^\prime_4 \nonumber\\
\gamma_2^\prime &\longleftrightarrow \gamma_3^\prime.
\label{eq:ontictrans}
\end{align}
Applying this to $\mu^\pm_Z$, it is clear that this the transformation
(\ref{eq:ontictrans})  doesn't   work,  implying  that   the  coherent
operation (\ref{eq:invertop}) is disallowed.

This observation prompts  the question of whether  an analogous result
can  be  produced  for   any  operational  theory  $\mathcal{T}$  with
non-simplicial space $\Sigma$ in  the $\Sigma$-ontology.  Here, in the
spirit  of  Spekkens'  argument,  it  may  be  assumed  that  coherent
transformations in the operational theory should have a representation
in  terms  of  reversible   transformations  among  the  ontic  states
$\lambda$  in the  ontological  model $\mathcal{T}_\sqcup$  or in  the
intermediate underlying  model in the g-type  ontology. Then universal
coherent operations can be shown not to exist.

To see this we note that  Spekkens' proof, in our terminology, employs
the  s-type ontology.   However, it  is readily  verified that  in the
g-type ontology, the following  ontic transformation indeed implements
the universal inverter (\ref{eq:invertop}):
\begin{align}
\gamma_1  &\longleftrightarrow   \gamma_8 \nonumber\\
\gamma_2 &\longleftrightarrow \gamma_7 \nonumber \\
\gamma_3 &\longleftrightarrow \gamma_6 \nonumber \\
\gamma_4 &\longleftrightarrow \gamma_5.
\label{eq:ontictrans0}
\end{align}
This shows that  this nonclassicality is dependent  on the ontological
model.  From the perspective of  interpretation, this result says that
if certain coherent operations are not observed, then it restricts the
class of possible underlying classical explanations.

Mathematically,  the result  shows that  by increasing  the number  of
ontic  states for  the same  operational  theory, we  can create  more
``space for maneurvre''.  Given a discrete $m$-input-$n$-output theory
$\mathcal{T}$, with space $\Sigma$, $|\partial\Sigma|=mn$ pure states,
and  dimension  $D_\mathcal{T}=m(n-1)$,  we have  for  the  underlying
simplex,  $D_\sqcup=n^m-1$   and  $|\partial\Sigma_\sqcup|=n^m$.   The
cardinality   $|\partial\Sigma|$   is   exponentially   smaller   than
$|\partial\Sigma_\sqcup|$. More  generally, $\partial\Sigma$  could be
arbitrarily enlarged  within a  fixed underlying  $\Sigma_\sqcup$. The
theory represented by the circular disk in Figure \ref{fig:JD} is such
an example.

There are  in all $|\partial\Sigma|!$ permutations  and hence distinct
coherent transformations in $\mathcal{T}$.   In the underlying theory,
there   are   $|\partial\Sigma_\sqcup|!$   ontic   permutations.    If
$|\partial\Sigma|!     >    |\partial\Sigma_\sqcup|!$,    and    hence
$|\partial\Sigma|  >  |\partial\Sigma_\sqcup|$,  then one  there  will
coherent operations that are  impossible in the $\Sigma$-ontology, but
whether such coherent operations are physically interesting is another
matter.

As an illustration,  we can show how to include  arbitrarily many pure
points   through  the   function  $\varphi_\star$,   to  create   such
impossibilities.   Given ontic  states (\ref{eq:ontixtate}),  from Eq.
(\ref{eq:GBC})  we  obtain  the ontic  probability  distributions  for
states in (\ref{eq:uneq}) to be:
\begin{eqnarray}
\mu_X^+   &=& (\half,\half,0,0), \nonumber \\
\mu_X^-   &=& (0,0,\half,\half),\nonumber \\
\mu_Z^+   &=& (\half,0,\half,0), \nonumber \\
\mu_Z^-   &=& (0,\half,0,\half).
\label{eq:dist}
\end{eqnarray}
The operational  theory characterized by  space $\Sigma$ which  is the
convex hull  of the states  (\ref{eq:uneq}) is a  non-simplex 
geometrically contained 
within the underlying gdit theory $\Sigma_\cup$.

This theory can be described as a nonsimplex in the convex hull of the
classical states (\ref{eq:ontixtate}).   From Eq.  (\ref{eq:dist}), we
see that the ontic operations
\begin{eqnarray}
\gamma_1 &\longleftrightarrow& \gamma_4, \nonumber\\
\gamma_2 &\longleftrightarrow& \gamma_3,
\label{eq:invertont}
\end{eqnarray}
indeed implement inversion (\ref{eq:invertop}).  

Suppose we add to $\mathcal{T}$ any state quantumly realizable through
a rotation with  angle $\theta$ in the $X-Z$  (equatorial) plane. This
is seen to be:
\begin{align}
|\psi_Y^+) &= (\delta, 1-\delta~|~\upsilon,1-\upsilon) \nonumber\\
|\psi_Y^-) &= (1-\delta, \delta~|~1-\upsilon,\upsilon)
\label{eq:upsi}
\end{align}
where   
\begin{align}
\delta   &\equiv  \frac{1}{2}(1+\sin\theta) \nonumber\\
\upsilon   &\equiv \cos^2(\frac{\theta}{2}).
\label{eq:delups}
\end{align}
These can be derived from the corresponding ontic distributions:
\begin{align}
\mu^+_Y &= (\delta\upsilon, \delta(1-\upsilon), 
(1-\delta)\upsilon, (1-\delta)(1-\upsilon)), \nonumber \\
\mu^-_Y &= ((1-\delta)(1-\upsilon), (1-\delta)\upsilon, 
\delta(1-\upsilon), \delta\upsilon),
\label{eq:laps}
\end{align}
where $0  \le \delta,  \upsilon \le  1$.  Clearly,  if the  added pure
points    have   the    pattern    (\ref{eq:laps}),   then    $\mu^+_Y
\longleftrightarrow     \mu^-_Y$     under    ontic     transformation
(\ref{eq:invertont}), making all  pure states are invertible. 

However, consider the coherent transformation defined by:
\begin{subequations}
\begin{eqnarray}
\psi_X^\pm &\longleftrightarrow \psi_Z^\pm \label{eq:arbcoha}\\
\psi_Y^+ &\longleftrightarrow \psi_Y^-. \label{eq:arbcohb}
\end{eqnarray}
\label{eq:arbcoh}
\end{subequations}
Eq. (\ref{eq:arbcoha}) requires $\gamma_1$  to transform to $\gamma_2$
or   $\gamma_3$,   neither   of    which   is   congruent   with   the
Eq.  (\ref{eq:arbcohb}), as  seen  in  Eq. (\ref{eq:laps}).   However,
clearly the operation (\ref{eq:arbcoh}) is rather articifical.

On the other hand, the existence of fewer pure states in the overlying
theory   $\mathcal{T}$  than   extreme   points   in  the   underlying
$\Sigma_\sqcup$  does  not  guarantee   that  all  universal  coherent
operations are possible.  To see this, suppose we replace the last two
states in Eq. (\ref{eq:spekkens}) by:
\begin{align}
|\psi_Z^+) &\equiv (\frac{1}{3}, \frac{2}{3} \mid \half, \half \mid 1,0) \nonumber\\
|\psi_Z^-) &\equiv (\frac{1}{3}, \frac{2}{3} \mid \half, \half \mid 0,1).
\label{eq:spekkens0}
\end{align}
In this case, the resulting theory remains nonsimplicial,
but the $Z$ eigenstates satisfy
$\half(\psi_Z^+   +   \psi_Z^-) = (\frac{1}{3},\frac{2}{3}\mid \half,\half
\mid \half, \half)$.
It follows that
\begin{align}
\mu_Z^+ &= \frac{1}{6}(\gamma_1+\gamma_3)+
\frac{2}{6}(\gamma_5+\gamma_7) \nonumber\\
\mu_Z^- &= \frac{1}{6}(\gamma_2+\gamma_4)+
\frac{2}{6}(\gamma_6+\gamma_8),
\label{eq:spekkensmu1}
\end{align}
instead of  the last  two equations  in (\ref{eq:spekkensmu}).   It is
observed that in  this case there is no ontic  permutation of the kind
(\ref{eq:ontictrans0}), implying the lack  of existence of a universal
inverter.

\subsection{Preparation contextuality
and nonsimpliciality \label{sec:prep}}

Suppose there  are different  preparations $P, P^\prime,  \cdots$ that
lead  to  the  same  outcome  statistics  in  the  operational  theory
$\mathcal{T}$  such  that  for   all  measurements  $M$,  $p(k|M,P)  =
p(k|M,P^\prime)  = \cdots$.   Thus, operationally,  these preparations
are indistinguishable  or equivalent.  Consider the  ontological model
for   these  preparations,   given  by   the  underlying   probability
distributions  $\mu_P(\lambda), \mu_{P^\prime}(\lambda),  \cdots$ over
ontic states $\lambda$.   If $\mu_P(\lambda) = \mu_{P^\prime}(\lambda)
= \cdots$,  i.e., they are also  ontologically indistinguishable, then
theory $\mathcal{T}$  is said to be  preparation non-contextual; else,
the theory is preparation contextual \cite{Spe05}.

Consider the regular theory determined by pure states (\ref{eq:uneq}).
It satisfies the condition (\ref{eq:psiX=Z}), in which the LHS and RHS
represent distinct  mixtures.  The  pure states  are described  by the
ontological  probability  distributions  $\mu_X^\pm$  and  $\mu_Z^\pm$
given  by  (\ref{eq:dist}).   Since  these  ontological  distributions
satisfy
\begin{align}
\frac{1}{2}(\mu_X^+(\lambda)  +   \mu_X^-(\lambda)) &=   
\frac{1}{2}(\mu_Z^+(\lambda)  +
\mu_Z^-(\lambda)),\nonumber\\
&= (\frac{1}{4},\frac{1}{4},\frac{1}{4},\frac{1}{4})\nonumber\\
&\equiv \mathcal{Y},
\label{eq:X=Z}
\end{align}
the  corresponding  ontic  mixtures  also  are  indistinguishable  and
therefore  preparation contextuality  cannot be  demonstrated in  this
case.

If the  $Y$-outcomes are  taken to be  given by  Eq.  (\ref{eq:upsi}),
then we have from Eqs.  (\ref{eq:laps}) and (\ref{eq:delups})
\begin{align}
\frac{1}{2}(\mu_Y^+(\lambda)  +   \mu_Y^-(\lambda))
&= \frac{1}{4}(1+\xi,1-\xi,1-\xi,1+\xi)\nonumber\\
&\ne \mathcal{Y},
\label{eq:Y!=Z0}
\end{align}
where $\xi \equiv \frac{\sin(2\theta)}{2}$,  giving a similar argument
for   preparation   contextuality.    Thus,  only   cases   $\theta=0$
(corresponding  to $X$)  and $\theta=\frac{\pi}{2}$  (corresponding to
$Z$) are the two exceptional, preparation noncontextual cases.

The above counterexamples to  preparation contextuality illustrate the
following general result.
\begin{thm}
A  nonclassical   theory  $\mathcal{T}$  is  in   general  preparation
contextual in the $\Sigma$-ontology.
\label{thm:prepcont}
\end{thm}
\begin{proof} 
Quite generally, let $\mathcal{T}$ be a regular theory embedded in the
gdit  theory $\mathcal{T}_\cup$,  which  is the  convex  hull of  pure
points $g_j$ ($j \in \{0,1,2,\cdots,N\})$.

A nonsimpliciality condition in $\mathcal{T}_\cup$ has the form:
\begin{equation}
\sum_{j=0}^M g_j = \sum_{k=M+1}^{N} g_k,
\label{eq:gnonsimplex}
\end{equation}
where $M+1<N$ and the LHS and RHS represent preparation-like contexts,
though  not proper  preparation  contexts because  gdit theories  lack
unambiguous preparability of pure states (cf. Section \ref{sec:gdit}).

Although  map   $\varphi_\triangledown$  is  non-injective,   the  map
$\varphi_\triangledown^{-1}$ is  one-to-one when  applied to  the pure
points  $g_j$'s  and  thus  well defined  for  the  domain  $\{g_j\}$.
However, upon inverting  the individual $g_j$'s in the LHS  and RHS of
Eq.  (\ref{eq:gnonsimplex}), we find by virtue of the simpliciality of
the underlying simplex $\Sigma_\sqcup$, that
\begin{equation}
\sum_{j=0}^M       \varphi_\triangledown^{-1}(g_j)      \ne       \sum_{k=M+1}^{N}
\varphi_\triangledown^{-1}(g_k),
\label{eq:ussimplex}
\end{equation}
which   is  a   manifestation  of   the  fact   that  $\varphi_\triangledown$   is
non-injective.     Inequality    (\ref{eq:ussimplex})     implies    a
preparation-like   contextuality,  since   it  demonstrates   the  two
preparation-like  contexts  in  the gdit  theory  are  distinguishable
ontologically.

Now consider  the nonsimpliciality  condition in the regular theory
$\mathcal{T}$:
\begin{equation}
\sum_j \psi_j = \sum_k \phi_k,
\label{eq:Tnonsimplex}
\end{equation}
where $\psi_j$  and $\phi_k$  are possibly un-normalized  basis states
from two different bases.  Eq. (\ref{eq:Tnonsimplex}) would be trivial
if both the LHS and RHS  independently sum to same quantity, typically
$\sum_{j=0}^{N} g_j$,  and thus (\ref{eq:Tnonsimplex}) would  not make
use of  the nonsimpliciality (\ref{eq:gnonsimplex}) of  the underlying
gdit theory.

Quite generally, the condition (\ref{eq:Tnonsimplex}) can be satisfied
by the requirement:
\begin{align}
\sum_i \psi_i &= f\left[\sum_{j=0}^M g_j\right] + 
(1-f)\left[\sum_{k=M+1}^{N} g_k\right] \nonumber\\
\sum_i \phi_i &= g\left[\sum_{j=0}^M g_j\right] + 
(1-g)\left[\sum_{k=M+1}^{N} g_k\right],
\label{eq:ek}
\end{align}
where $0 \le  f,g \le 1$ and  it is not necessary that  $f=g$. Then we
derive  (\ref{eq:Tnonsimplex}) nontrivially  because in  this case  we
must  make  use  of  Eq.   (\ref{eq:gnonsimplex}).  The  trivial  case
corresponds to $f=g$.

Mapping  the two  LHS in  (\ref{eq:ek}) to  the underlying  simplex by
applying $\varphi^{-1}_\triangledown$  to the $g_j$'s we  obtain the corresponding
mixtures in the underlying simplex $\Sigma_\sqcup$:
\begin{align}
\varphi_\triangledown^{-1}: \sum_i \psi_i &\mapsto 
f\left[\sum_{j=0}^M \varphi_\triangledown^{-1}(g_j)\right] + (1-f)\left[\sum_{k=M+1}^{M+N} 
\varphi_\triangledown^{-1}(g_k)\right]  \nonumber\\
& \equiv \varphi^{-1}_\triangledown(\Psi) \nonumber\\
\varphi_\triangledown^{-1}: \sum_i \phi_i &\mapsto 
g\left[\sum_{j=0}^M \varphi_\triangledown^{-1}(g_j)\right] + (1-g)\left[\sum_{k=M+1}^{M+N} 
\varphi_\triangledown^{-1}(g_k)\right] \nonumber\\
&\equiv \varphi^{-1}_\triangledown(\Phi).
\label{eq:ekh}
\end{align}
However, in view of Eq. (\ref{eq:ussimplex}), it follows that
\begin{equation}
\varphi^{-1}_\triangledown(\Psi) \ne \varphi^{-1}_\triangledown(\Phi), 
\end{equation}
entailing that two distinct underlying  states map to the same mixture
in  the  overlying  theory $\mathcal{T}$,  which  implies  preparation
contextuality.  \hfill \qed
\end{proof}

As   a   simple   illustration   of   preparation   contextuality   in
$\Sigma$-ontology, note that in the  case of the 2-input-2-output gdit
theory, the  nonsimpliciality condition among the  four extreme points
$(j,k)$, where $j,k\in\{0,1\}$, is given by:
\begin{equation}
\half[(0,0) + (1,1)] = \half[(0,1) + (1,0)].
\end{equation}
Under $\varphi_\triangledown^{-1}$, the LHS maps to the mixture $(1,0,0,1)$ in the
underlying   simplex  $\Sigma_\sqcup$,   whilst   the   RHS  maps   to
$(0,1,1,0)$. 

According  to   Theorem  \ref{thm:prepcont},  we   obtain  preparation
contextuality  for  bases   in  which  the  uniform   mixture  of  the
eigenstates should  correspond to an ontic  mixture in $\Sigma_\sqcup$
of  the form  $(a,b,b,a)$ where  $2(a +  b)=1$ while  $a\ne b$.   This
explains why we  fail to obtain a proof  for preparation contextuality
using   the    measurements   $X$    and   $Z$,   as    determined   by
Eq. (\ref{eq:uneq}), in view of (\ref{eq:X=Z}).  On the other hand, by
a  similar  argument, any  two  measurements  $Y$ parametrized  by  two
distinct    values    of    $\theta    \in    [0,\frac{\pi}{2})$    in
  Eq. (\ref{eq:laps}) yield an  argument for preparation contextuality
  by virtue of Eq.  (\ref{eq:Y!=Z0}).   

All  the above nonclassical  features, whether with  reference to
the ontology or not, require only measurement incongruence, without
reference to the question of  the transitivity of congruence or its
extendability to higher orders.  These  features correspond to what we
call as  base level  nonclassicality.  In  the following  sections, we
shall consider  nonclassical phenomena  arising when  congruence is
not  transitive or  extendable  to higher  orders,  and correspond  to
so-called ``higher-level'' nonclassicality.

\section{From incongruence to contextuality \label{sec:JD}}

If   pairwise   congruence  is   transitive   in   a  regular   theory
$\mathcal{T}$,  then   the  above   theorems  (discussed   in  Section
\ref{sec:conseq}) suffice  to characterize the nonclassicality  of the
theory.  Further, the geometro-ontological models discussed previously
(in Section \ref{sec:simplionto}) also apply.

However,  if  pairwise  congruence   is  intransitive  in  general  in
$\mathcal{T}$,  then   any  definite  value  assignment   to  all  the
measurements in question can't  be context-independent.  Therefore, the
above  ontological models,  which presume  a JD  for all  measurements,
cannot  be applied  without  suitable generalization.

Now assume that  congruence is extendable, whereby if $A$  and $B$ are
congruent and so are  $B$ and $C$, then it is assumed  that $A, B$ and
$C$ are jointly  measruable, and so on for any  number of measurements.
Later we will discuss examples of how extendability can break down.

\begin{thm}
If a regular  operational theory $\mathcal{T}$ in  which congruence is
transitive and  extendable, then  all states  in $\mathcal{T}$  have a
joint distribution (JD) in  an outcome-deterministic underlying model.
However, the converse is not true (i.e., intransitivity does not imply
lack of JD).
\label{thm:nary}
\end{thm}
\begin{proof} 
The classical case,  where all measurements are  pairwise congruent, is
trivial.  For the case where no  two variables are congruent (and thus
congruence  is  trivially transitive),  given  measurements  $X, Y,  Z,
\cdots$,  one can  always construct  the JD  in an  underlying theory,
given  by $P(x,y,z,\cdots)  {:=}  P(x)P(y)P(z)\cdots$.   For the  case
where  incongruence  exists  but  congruence  is  transitive,  we  can
partition  all measurements  into equivalence  classes $\mathcal{Q}_j$,
whereby all  measurements $Q_j^{(k)}$  within the class  are congruent,
and any two measurements in  different classes, such as $Q_m^{(j)}$ and
$Q_n^{(k)}$ $(m\ne  n)$ are incongruent.   All measurements in  a given
equivalence  class   are  jointly  congruent  and   therefore  jointly
measurable, by the assumption of extendability.

We  can  thus  treat  each  equivalence  class  as  a  single  ``class
measurement'' $\mathcal{Q}_j$ determined  by a probability distribution
$P^Q_j$.  Therefore  we can always  write down  a JD in  an underlying
(outcome deterministic) model for all measurements, given by:
\begin{equation}
P(q^{(1)}_1,q^{(2)}_2,\cdots,q^{(1)}_2,q^{(2)}_2,\cdots) = P^Q_1P^Q_2\cdots 
\label{eq:QJD}
\end{equation}
where  the LHS  denotes the  probability over  all measurements  in the
theory. 

To see that  the converse is not true, consider  the theory (fragment)
$\mathcal{T}_{2,2}$  consisting  of   four  measurements,  $\mathcal{A}
\equiv \{A_1, A_2\}$ and $\mathcal{B}  \equiv \{B_1, B_2\}$, such that
any  element  of  $\mathcal{A}$  is  congruent  with  any  element  of
$\mathcal{B}$,  i.e.,  any pair  $(A_j,B_k)$  is  congruent, which  we
represent by $R(A_j,B_k)$ where $R$ represents the congruence relation
template.   Now suppose  we have  $\neg R(A_1,A_2)$  but $R(B_1,B_2)$,
i.e., incongruence on only one side.  We then have intransitivity from
the            chains           $R(A_1,B_1),R(B_1,A_2)$            and
$R(A_1,B_2),R(B_2,A_2)$. However,  this does not  lead to lack  of JD,
because we  can construct  a JD  based on the  quantities that  can be
jointly measured:
\begin{equation}
P(A_1,A_2,B_1,B_2) = \frac{P(A_1,B_1,B_2)P(A_2,B_1,B_2)}{P(B_1,B_2)},
\label{eq:6gd}
\end{equation}
which has  been constructed such  that tracing  out $A_1$ or  $A_2$ or
$A_1,A_2$ returns the operational  probabilities. \hfill \qed
\end{proof}

In light of this result, the existence of incompatibility is necessary
but  not   sufficient  for  a  Kochen-Specker   theorem  contextuality
\cite{KS67}.  A  non-trivial super-transitive structure  of congruence
is required to thwart JD.

In respect  of the above  example, note  that if we  have incongruence
among both the $A_j$'s and  $B_k$'s, i.e., $\neg R(A_1,A_2)$ and $\neg
R(B_1,B_2)$, then it can be shown  that JD does not exist.  Therefore,
in such a  $\{\mathcal{A}, \mathcal{B}\}$ scenario, we  obtain a tight
link between congruence and JD, which clarifies such a link studied by
various authors \cite{BGG+0, *Ban15, SAB+05, WGC09, QVB14, SB14}.

\subsection{Unextendability of congruence \label{sec:hic}}

We  now  point out  how  congruence  of  measurements  can fail  to  be
extendable.   An example  is  provided  by Specker's  ``overprotective
seer'' (OS)  correlations \cite{LSW11}.   Here let $A,  B$ and  $C$ be
three dichotomic measurements, such that  any two are congruent and can
be  measured,  and the  outcomes  will  be anticorrelated  with  equal
probability.  However, they can't be jointly measured because there is
no JD over  measurements $A,B,C$ for the  state. If it did,  it must be
probability  distribution   of  the   8  three-bit  string   $abc  \in
\{+1,-1\}^3$. Anticorrelation on $AB$ and  on $BC$ implies $abc$ is of
the  pattern  $+-+$ or  $-+-$,  but  in this  case  $AC$  will not  be
anticorrelated.  The quantity $\langle AB\rangle + \langle BC\rangle +
\langle  AC\rangle$,   which  is  the  OS   inequality,  thus  reaches
noncontextuality the minimum of -1.

We can now extend Specker's idea  to an ``extended OS'' (XOS) theory .
Consider  the state  $\rho$ in  an operational  theory, which  has the
property that for any threeway joint-measurement, all three outcomes
must be distinct.  There is no  JD over measurement outcomes $a, b, c$
and $d$ that can satisfy this.  To  see this suppose that we assign 0,
1, 2 to  $A, B, C$.  Then to satisfy  the distinctness requirement for
$B, C, D$, we should assign 0 to $D$.  However, this assignment scheme
will imply that measuring $A, B, D$ would yield $0,1,0$, violating the
outcome distinctness requirement.

This  contextual  correlation  in  the   extended  OS  theory  can  be
experimentally tested via the violation of the inequality:
\begin{equation}
\langle \overline{ABC}\rangle +
\langle \overline{BCD}\rangle +
\langle \overline{CDA}\rangle +
\langle \overline{DAB}\rangle \le 2,
\label{eq:XOS}
\end{equation}
where  each  of the  four  summands  is experimentally  determined  as
follows.

For any  three jointly measurements  (say $A, B, C$),  in each
trial determine the maximum, intermediate  and minimum outcomes of the
three  measurements, which  are  denoted $v_{\max},  v_{\rm mid}$  and
$v_{\min}$,  respectively, provided  the three  are distinct.   If the
three outcomes are  not distinct but one of them  is distinct from the
other two,  then assign $v_{\max}$  (resp., $v_{\min}$) to  the larger
(resp., smaller)  outcome and assign  $v_{\rm mid}$ to the  outcome is
majority. (E.g., if  measuring $A, B, C$ produces  outcomes 0,1,1 then
assign 0 to $v_{\min}$ and 1 to both $v_{\rm mid}$ and $v_{\max}$.) If
all three  outcomes of  a three-way joint measurement  are identical,
then  assign  that value  uniformly  to  $v_{\max}, v_{\rm  mid}$  and
$v_{\min}$.

For any three measurements (say $A, B, C$), the quantity 
\begin{equation}
\overline{ABC}        \equiv        \langle       (v_{\max}        -
v_{\rm mid})(v_{\rm mid}-v_{\min})  \rangle,
\end{equation}
where the  angles indicate the  average over  a number of  trials.  In
each   trial,   the   quantity  $(v_{\max}   -   v_{\rm   mid})(v_{\rm
  mid}-v_{\min})$  is  1 only  if  the  outcomes are  distinct  (i.e.,
0,1,2), and zero otherwise (e.g., 0,1,1 or 1,1,1).  The other summands
in  Eq.   (\ref{eq:XOS})   are  calculated  in  this   way.   For  any
noncontextual value assignment of $A, B, C, D$, it may be checked that
the LHS is bounded above by 2.  On the other hand, for XOS theory, the
LHS of (\ref{eq:XOS}) attains the algebraic maximum of 4.

The above consideration about other  possible structures prompts us to
reflect on  why compatibility in  QM doesn't seem to  be unextendable.
Here we propose  that this is ``the most natural''  for a nonclassical
theory.  The reason is that otherwise, nonseparability would intervene
mysteriously  at some  higher despite  compatibility at  lower orders.
This would mean  that nonseparability does not occur  as a consequence
of  failure  of  extendability  of  congruence.   Thus,  the  ``nice''
connection to intransitivity  of congruence would be  lost, making the
theory   more   complicated.   From   this   perspective,   QM  is   a
\textit{natural}  nonclassical  theory,   which  provides  a  possible
explanation for the absence of OS type correlations in QM.

\section{Ontology for contextual theories \label{sec:cbox}}

From the perspective of  $\Sigma$-ontology, basic nonclassicality of a
regular theory arises essentially because of the compressive action of
$\varphi_\triangledown$ on the underlying simplex $\Sigma_\sqcup$ (the
map  $\varphi_\star$   adds  attractive   features  like   pure  state
preparability).   The  simplex  $\Sigma_\sqcup$ itself  has  a  direct
product  structure  and  is   quite  classical.   In  particular,  the
underlying versions of the measurements are separable, meaning that any
state $\psi$ has  a JD over the values of  the measurements.  This idea
forms  the basis  of the  underlying simplex  $\Sigma_\sqcup$.  Higher
nonclassicality  will be  seen  to  correspond to  the  addition of  a
further  fundamental type  of nonclassicality,  which occurs  when the
underlying simplex  itself deviates from the  direct product structure
form, leading to quantum contextuality. 

When a  contextuality inequality is  violated and hence no  JD exists,
the underlying simplex $\Sigma_\sqcup$  of the type described earlier,
which assumes  separable measurements, does  not exist.  One  can still
envisage an underlying simplex, provided  we enlarge the convex direct
product suitably.   We shall  find that the  \textit{convex contextual
  product},   denoted  by   the  symbol   $\votimes$,  fulfills   this
requirement.   Essentially, the  operation $\votimes$  is a  recursive
application of  $\uotimes$ to combinations of  measurements.  This will
be      used      to      define      an      underlying      simplex,
$\Sigma_\sqcup^{\theta\bullet}$,  suitable  for  contextuality.   Here
$\theta$ is  the ``valency'' of  an measurement, meaning the  number of
other  measurements  it  must  combine  with,  to  determine  the  full
measurement  context.   
% In fact,  larger valencies  correspond to  higher-order interference
% than   observed  in   a   doubles-slit  interferometric   experiment
% \cite{Sor94}.
We  shall  sometimes  indicate  the  valency by  a  subscript  to  the
$\votimes$ operator.

Suppose  $\theta=1$, so  that the  context  is specified  by a  single
fiducial measurement. Given fiducial measurements  $X, Y, Z, \cdots$, we
define
\begin{equation}
\Sigma_X \votime \Sigma_Y \votime \Sigma_Z \votime
\cdots \equiv \Sigma_{XY} \uotimes \Sigma_{YZ}
\uotimes \Sigma_{XZ} \uotimes \cdots,
\label{eq:votimes}
\end{equation}
where  $\Sigma_{XY} =  \Sigma_X \uotimes  \Sigma_Y$, and  so on.   For
$\theta=2$,  the RHS  will  contain operands  such as  $\Sigma_{XYZ}$,
meaning that the  context is specified by two  measurements, an example
of which is studied later.  Here it suffices to note that, as with the
convex  direct  product,  if  the   convex  sets  being  combined  are
simplexes,  then  the convex  contextual  product  is a  simplex.   In
general, the  marginal probabilities $P_X,  P_Y, P_Z, \cdots$  are not
well defined,  which may  be averted  by imposing  additional backward
consistency conditions.   In the absence of  backward consistency, the
full combined space  will be of the  form $\Sigma_{XYZ\cdots} \uotimes
\Sigma_{XYZ\cdots}^{(1)}  \uotimes  \Sigma_{XYZ\cdots}^{(2)}  \cdots$,
i.e., probabilities are explicitly specified for different valencies.

In      particular,      the      vertices     of      convex      set
$\Sigma_\sqcup^{\theta\bullet}$     are    in     general    classical
(deterministic) contextual configurations, meaning, they correspond to
a  classical   mechanism  to   reproduce  contextual   behavior.   The
underlying     simplex     $\Sigma_\sqcup$introduced    in     Section
\ref{sec:simplionto} corresponds to $\theta=0$.

An  example  for  a  classical contextual  configuration  suitable  to
simulate the OS  theory would be as follows: the  properties $X, Y, Z$
have definite  values 0, 0,  1, except that  if $X$ is  measured along
with  $Y$,  then  a  hidden   mechanism  causes  a  ``signal''  to  be
transmitted from  $Y$ to $X$  instructing $X$  to flip its  bit value.
This may be represented as follows: {\small
\begin{equation}
\begin{array}{|c|c|}
\hline
{\rm input} & {\rm output} \\
\hline
XY & 10\\
\hline
\end{array}
\otimes
\begin{array}{|c|c|}
\hline
{\rm input} & {\rm output} \\
\hline
YZ & 01\\
\hline
\end{array}
\otimes
\begin{array}{|c|c|}
\hline
{\rm input} & {\rm output} \\
\hline
ZX & 10\\
\hline
\end{array}
\label{eq:ccc}
\end{equation}} 
For  the  OS  theory,  there  are  $4^3  =  64$  classical  contextual
configurations such  as (\ref{eq:ccc}), which constitute  the vertices
of  a   63-simplex,  denoted  $\Sigma_\sqcup^{1\bullet}$,   or  simply
$\Sigma_\sqcup^{\bullet}$, given by:
\begin{equation}
\Sigma_\sqcup^\bullet = \Sigma_X ~\overline{\otimes}~ \Sigma_Y
~\overline{\otimes}~ \Sigma_Z.
\label{eq:otensor}
\end{equation}
Our strategy will be to use  this object as the underlying simplex and
impose backward consistency  on the gdit space  derived by application
of the compression operation $\varphi_\triangledown$.

In  g-type   ontology,  the   action  of   $\varphi_\triangledown$  on
$\Sigma_\sqcup^\bullet$ is a transformation analogous to that of going
from   the   form   of    Eq.   (\ref{eq:tensorprod})   to   that   of
Eq.   (\ref{eq:tensorsum}).   Under   this   mapping,  the   classical
contextual configuration Eq.  (\ref{eq:ccc}) yields {\small
\begin{equation}
\begin{array}{|c|c|}
\hline
{\rm input} & {\rm output} \\
\hline
XY & 10\\
\hline
\end{array}
\oplus
\begin{array}{|c|c|}
\hline
{\rm input} & {\rm output} \\
\hline
YZ & 01\\
\hline
\end{array}
\oplus
\begin{array}{|c|c|}
\hline
{\rm input} & {\rm output} \\
\hline
ZX & 10\\
\hline
\end{array},
\end{equation}}
which can be represented more conventionally as:
\begin{equation}
\begin{array}{c|c}
\hline
\textrm{input} ~&~ \textrm{output} \\
\hline 
XY & 10 \\
\hline
YZ & 01 \\
\hline
ZX & 10 \\
\hline
\end{array}
\label{eq:cgdit}
\end{equation}
the \textit{contextual  gdit}, since it  is the gdit version  of these
classical    deterministic    configurations.     Since    compression
$\varphi_\triangledown$  preserves the  extreme points,  there are  64
contextual gdits, whose  convex hull is the  corresponding gdit theory
space $\Sigma_\cup^{\theta\bullet}$, the \textit{underlying contextual
  gdit space}, with  dimension $3 \times (4-1)=9$  (a 7-fold reduction
in dimensionality).

This  contextual gdit  violates  the  ``no-signaling'' conditions  for
single   systems,   more   precisely,   \textit{non-contextuality   of
  probabilities} in the sense of Gleason's theoreom \cite{Sor94,G57}:
\begin{align}
P(a|AB)  &=  P(a|AC),\nonumber\\ 
P(b|BA)  &=  P(b|BC),\nonumber\\ 
P(c|CA)  &= P(c|CB),
\label{eq:Gleason}
\end{align}
where $P(a|AB)$ denotes the probability to obtain outcome $a$ when $A$
is measured alongside $B$.

It  is  convenient to  compose  the  crumpling  map  in terms  of  two
sub-mappings: $\varphi_\star = \varphi_\star^- \circ \varphi_\star^+$.
The   mapping    $\varphi_\star^+:   \Sigma_\cup^\bullet   \rightarrow
\Sigma^+$   imposes   the   ``contextual   no-signaling''   conditions
(\ref{eq:Gleason}), where $\Sigma^+$  refers to the \textit{contextual
  no-signaling  polytope}, i.e.,  the largest  convex set  embedded in
$\Sigma_\cup^\bullet$   consistent   with   contextual   no-signaling.
Obviously, $\Sigma \subseteq  \Sigma^+$, and mapping $\varphi_\star^-:
\Sigma^+ \rightarrow \Sigma$.

For a  theory of 3-input-$n$-outcome systems,  Eq.  (\ref{eq:Gleason})
would impose $N_G  \equiv 3(n-1)$ conditions, whilst  the dimension of
the contextual  gdit theory $\Sigma_\cup^{\bullet}$ underlying  the OS
theory would  be $D_\cup^{\bullet} \equiv 3(n^2-1)$.   Therefore, with
imposition  of conditions  (\ref{eq:Gleason}),  the  dimension of  the
contextual no-signaling polytope $\Sigma^+$,  embedded within the gdit
polytope  $\Sigma_\cup^{\bullet}$ is,  $D_\cup^{\bullet}-N_G=3n(n-1)$.
By   tomographic  separability,   the   nonsimplex   $\Sigma$  for   a
3-input-$n$-output regular  theory has dimension  $3(n-1)$.  Therefore
the  map $\varphi_\star^-:  \Sigma^+  \rightarrow  \Sigma$ involves  a
further  $n$-fold reduction  of  dimensionality (for  a discussion  on
dimensionality mismatch in a different context, cf. \cite{BKL+14}).

\begin{example}
The following six contextual OS gdits
\begin{equation}
\begin{array}{c||c|c|c|c|c|c}
\toprule
 \multirow{ 2}{*}{\rm Input} & \multicolumn{6}{c}{\rm Output} \\
\cline{2-7}
 & \mathcal{Q}_1^A & \mathcal{Q}_2^A 
 & \mathcal{Q}_1^B & \mathcal{Q}_2^B &
\mathcal{Q}_1^C & \mathcal{Q}_2^C \\
\hline
AB & 01 & 10 & 01 & 10 & 01 & 10 \\
BC & 10 & 01 & 01 & 10 & 10 & 01 \\
CA & 01 & 10 & 10 & 10 & 10 & 01 \\
\hline
\end{array},
\label{tab:os}
\end{equation}
each of  which violates the  OS inequality, by reaching  the algebraic
minimum of $-3$.  In Eq.  (\ref{tab:os}), the first column is the pair
of joint inputs, and each other  column is the output corresponding to
a given deterministic contextual  box $\mathcal{Q}_j^k$.  Here the box
$\mathcal{Q}_2^A$     is    just     the    gdit     represented    in
Eq. (\ref{eq:cgdit}).   The superscript corresponds to  the measurement
for which  a Kochen-Specker contradiction  would arise in  a realistic
assignment of values.

The  vertices  of  the  contextual  no-signaling  polytope  $\Sigma^+$
derived from the above contextual gdit polytope has vertices formed by
equally mixing a gdit like  $\mathcal{Q}^j_1$ and its ``partner'' gdit
$\mathcal{Q}^j_2$, where  the latter  is obtained  from the  former by
flipping the outcome bits, and $j \in \{A, B, C\}$. For the contextual
gdits   in   (\ref{tab:os}),   we  obtain   the   three   nonsignaling
\textit{contextual         boxes}        $\mathcal{Q}^j         \equiv
\half\left(\mathcal{Q}^j_1  +  \mathcal{Q}^j_2\right)$, which  violate
Eq.  (\ref{eq:XOS}) by reaching its algebraic minimum, but respect the
Gleason no-signaling conditions (\ref{eq:Gleason}).

Any state in  the polytope generated by  these contextual no-signaling
boxes  $\mathcal{Q}^j$  has  uncertainty  $\mathcal{U}=\half$  as  per
(\ref{eq:unc}),  where  the  uncertainty arises  from  the  underlying
mixing.   We  observe that  the  no-signaling  feature is  essentially
equivalent  to backward  consistency.   Therefore, any  such state  in
$\Sigma^+$  is   backward  consistent,   in  contrast  to   states  in
$\Sigma^\bullet_\cup - \Sigma^+$, in  particular, the contextual gdits
$\mathcal{Q}^j_k$.   The   $\mathcal{Q}^j$'s  are   the  contextuality
analogues of  the PR boxes  \cite{Sca06,PR94}, which violate  the CHSH
inequality for  quantum nonlocality to its  algebraic maximum.  \hfill
\qed
\end{example}

The complementarity between signaling and randomness in the context of
nonlocality   \cite{ASQIC},   based   on   techniques   developed   in
\cite{Pir03}, and  first conjectured  in \cite{KGB+11,Hal10b},  can be
formulated   analogously   for  contextuality,   reported   elsewhere.
Similarly,   one   may   consider    the   single-system   analog   of
device-independent  quantum cryptography  \cite{BB84,CK78,VV14,BHK05},
based on contextuality inequalities \cite{KCB+08}.

As  in  the case  of  OS  theory,  the underlying  contextual  simplex
$\Sigma_\sqcup^{2\bullet}$ for the XOS  theory can be constructed with
vertices given  by classical signaling configurations.   Under the map
$\varphi_\triangledown$,   these  configurations   transform  to   XOS
contextual gdits, which are the  vertices of the underlying contextual
gdit polytope, $\Sigma^{2\bullet}_\cup$.  Three such contextual gdits,
which are the XOS equivalents of (\ref{tab:os}), are:
\begin{equation}
\begin{array}{c||c|c|c}
\toprule
\multirow{2}{*}{\rm Input} & \multicolumn{3}{c}{\rm Output} \\
\cline{2-4}
 ~ & \mathcal{R}_1 & \mathcal{R}_2 & \mathcal{R}_3 \\
\hline
XYZ & 012 & 120 & 201 \\
\hline
YZW & 120 & 201 & 012 \\
\hline
ZWX & 201 & 012 & 120 \\
\hline
WXY & 012 & 120 & 201 \\
\hline
\end{array},
\end{equation}
each of which violates the XOS contextuality inequality (\ref{eq:XOS})
to its algebraic maximum of 4.  As in the OS case, each contextual XOS
gdit  $\mathcal{R}_j$ violates contextual no-signaling.   

Under  the  crumpling map  $\varphi_\star^+:  \Sigma_\sqcup^{2\bullet}
\rightarrow \Sigma^+$  we obtain the contextual  no-signaling polytope
$\Sigma^+$ appropriate to  the XOS toy theory,  whose vertices include
contextual boxes, one of which is obtained by taking a uniform average
$\frac{1}{3}      \left(\mathcal{R}_1      +      \mathcal{R}_2      +
\mathcal{R}_3\right)$.

Finally,  it is  worth contrasting  our combinatoric  approach to
contextuality   from    that   of    \cite{CSW14,   abramsky2011sheaf,
  silva2017graph},   where   graph    theoretic   representations   of
contextuality  and nonlocality  are  explored, which  we mentioned  in
Section  \ref{sec:simplionto}.   One key  difference  is  that in  our
approach,  the   adjacent  vertices   of  the  graph   representing  a
contextuality  situation  are   compatible  measurements,  whereas  in
\cite{CSW14}  they are  exclusive  events associated  with  a pair  of
measurements and  in \cite{abramsky2011sheaf} with  orthogonal vectors
in  $\mathbb{R}^d$.   The  difference   in  these  approaches  may  be
attributed to the different motivations. For example, in \cite{CSW14},
the authors aim to uncover a fundamental axiom to explain the level of
quantum   violation  of   contextuality   inequalities,  whereas   our
motivation here has been to  interpret contextuality as a higher-level
manifestation  of nonclassicality,  that  is fundamentally  associated
with the nonsimplicial character.

\section{Conclusions \label{sec:conclu}}

For  a class  of general  probability theories  of single  systems, we
point  out  that  a  number of  nonclassical  features  like  multiple
pure-state decomposability,  measurement disturbance, lack  of certain
universal operations,  measurement uncertainty and no-cloning,  can be
accounted  for  by the  assumption  that  the  state  space is  not  a
simplex. These nonclassical  features can be explained in  terms of an
epistemic interpretation for which an  underlying simplex serving as a
noncontextual  ontological model.   A  base-level nonclassical  theory
would be nonclassical in the  sense of possessing a no-cloning theorem
etc., but would lack KS contextuality.  An example of such a theory is
Spekkens  toy theory  \cite{Spe07}.   For  contextuality, the  further
assumption about the intrasitivity of congruence is necessary.

In  Spekkens'  toy  theory,  a single  system  has  three  measurements
(``questions'').  All  equal mixtures  of maximal knowledge  states of
the  three measurements  are  operationally indistinguishable,  meaning
that the state space is non-simplicial, though in that work it was not
described  geometrically,  but  instead  as  a  manifestation  of  the
``knowledge-balance''  principle.  On  the  other hand,  all of  three
measurements  are  pairwise incongruent,  so  that  the possibility  of
intransitivity of pairwise congruence does not arise.  Our result thus
explains why the theory reproduces many nonclassical features, but not
KS contextuality.   Therefore, in response to  Spekkens' open question
\cite[Section  IX]{Spe07}, as  to  what other  principle, besides  the
knowledge-balance  principle, is  required  to  capture the  remaining
quantum phenomena,  our approach suggests: namely,  the intransitivity
or unextendability of pairwise congruence.

Here we only consider single  systems, without reference to correlated
systems, in  contrast to various current  foundational studies.  Thus,
in the absence of other assumptions, all nonclassical features derived
in our study are logically independent of whether or not the theory in
question is nonlocal or nonsignaling.

By identifying  nonclassical features  that are logically  related and
those  that  are logically  independent  in  any regular  theory,  our
approach could provide a new perspective for axiomatic reconstructions
of      QM       \cite{      Dar07,      chiribella2016quantum,
  chiribella2017quantum,MM11, DB11, BMU14}.

In QM, various nonclassical features  are tied closely together by the
single theme of noncommutative algebra, making it difficult to discern
basic   interdependence  from   independence.    Our  approach   helps
disentangle the relationship  between different nonclassical features.
For instance, it shows that any nonclassical theory $\mathcal{T}$ with
a  no-cloning   theorem  is  also  expected   to  possess  measurement
uncertainty and measurement disturbance.   Thus, these latter features
do not require any ``further explanation''.

Finally,  we  believe  that  the explanatory  value  of  an  axiomatic
approach is enhanced  by the choice of axioms  that appear ``natural''
and  that are  naturally connected  with each  other.  Thus,  once one
isolates the most elementary surprise about a theory and identifies it
as  the  basic axiom,  this  should  lead  us  smoothly to  the  other
axiom(s).   In our  work, it  may be  hoped that  pairwise measurement
incongruence serves as such a basic axiom for QM, which then naturally
leads  to the  question of  (in)congruence among  multiple properties,
leading to contextuality.

The requirement  of naturalness of  axioms is not  entirely aesthetic,
but based  on our perception that  fulfilling it would position  us in
the state of ``the right ignorance''  that would guide us to the right
questions answering  which would  enable us to  understand why  QM has
this particular mathematical structure.

\begin{acknowledgments}
SA  acknowledges support  through Manipal  university and  the INSPIRE
fellowship  [IF120025] by  the Department  of Science  and Technology,
Govt.   of India.   AP  thanks Department  of  Science and  Technology
(DST),  India for  the  support provided  through  the project  number
EMR/2015/000393.
\end{acknowledgments}

% \bibliography{preeti}
\bibliography{thesis,preeti}

\begin{thebibliography}{73}
\bibitem [{\citenamefont {Plavala}()}]{Pla16}%
  \BibitemOpen
  \bibfield  {author} {\bibinfo {author} {\bibfnamefont {M.}~\bibnamefont
  {Plavala}},\ }\href@noop {} {\enquote {\bibinfo {title} {All measurements in
  a probabilistic theory are compatible if and only if the state space is a
  simplex},}\ }\bibinfo {note} {ArXiv:1608.05614}\BibitemShut {NoStop}%
\end{thebibliography}
 \end{document}